\documentclass[journal]{IEEEtran}
\usepackage{amsmath}
\usepackage{amssymb}
\usepackage{amsfonts}
\usepackage{graphicx}
\usepackage[width=18.5cm,height=24.5cm]{geometry}
\usepackage{enumerate,color,graphicx,fancybox,pifont,epsf,epsfig,subfigure,amsmath,amssymb,psfrag,amsthm}
\usepackage{algorithm}
\usepackage{algorithmic}
\usepackage{cite}
\newcounter{MYtempeqncnt}

\newtheorem{ex}{\textbf{Example}}
\newtheorem{theorem}{\textbf{Theorem}}

\newtheorem{proposition}[theorem]{\textbf{Proposition}}

\newcommand{\secref}[1]{Section~\ref{#1}}

\newcommand{\figref}[1]{Figure~\ref{#1}}

\newcommand{\proref}[1]{Proposition~\ref{#1}}

\newcommand{\remref}[1]{Remark~\ref{#1}}

\newcommand{\algref}[1]{Algorithm~\ref{#1}}
\newcommand{\exampref}[1]{Example~\ref{#1}}

\newtheorem{rem}{\textbf{Remark}}

\title{Joint Source-Channel Vector Quantization for \\ Compressed Sensing}
\author{Amirpasha Shirazinia, \textit{Student Member, IEEE}, Saikat Chatterjee, \textit{Member, IEEE}, Mikael Skoglund, \textit{Senior Member, IEEE} \\
}

\begin{document}
\maketitle

\begin{abstract}
We study joint source-channel coding (JSCC) of compressed sensing (CS) measurements using vector quantizer (VQ). We develop a framework for realizing optimum JSCC schemes that enable encoding and transmitting CS measurements of a sparse source over discrete memoryless channels, and decoding the sparse source signal. For this purpose, the optimal design of encoder-decoder pair of a VQ is considered, where the optimality is addressed by minimizing end-to-end mean square error (MSE). We derive a theoretical lower-bound on the MSE performance, and propose a practical encoder-decoder design through an iterative algorithm. The resulting coding scheme is referred to as channel-optimized VQ for CS, coined COVQ-CS. In order to address the encoding complexity issue of the COVQ-CS, we propose to use a structured quantizer, namely low complexity multi-stage VQ (MSVQ). We derive new encoding and decoding conditions for the MSVQ, and then propose a practical encoder-decoder design algorithm referred to as channel-optimized MSVQ for CS, coined COMSVQ-CS. Through simulation studies, we compare the proposed schemes vis-a-vis relevant quantizers. 
\end{abstract}

\section{Introduction} \label{sec:intro}

Compressed sensing (CS) \cite{08:Candes} considers retrieving a \textit{high-dimensional} sparse vector $\mathbf{X}$ from relatively \textit{lower} number of measurements. In many practical applications, the collected measurements at a CS sensor node need to be \textit{encoded} using finite bits and \textit{transmitted} over \textit{noisy} communication \textit{channels}. To do so, efficient design of source and channel codes should be considered for reliable transmission of the CS measurements over noisy channels. The optimum performance theoretically attainable in a point-to-point memoryless channel can be achieved using separate design of source and channel codes, but this performance requires infinite source and channel code block lengths resulting in delay as well as coding complexity. Considering finite-length sparse source and CS measurement vector, it is theoretically guaranteed that joint source-channel coding (JSCC) can provide better performance than a separate design of source and channel codes. Therefore, to design a practical coding method, we focus on optimal JSCC principles for CS in the current work. \textit{Denoting the reconstruction vector by $\widehat{\mathbf{X}}$ at a decoder, our main objective is to develop a generic framework for optimum JSCC of CS measurements using vector quantization, or in other words, optimum joint source-channel vector quantization for CS, such that $\mathbb{E}[\|\mathbf{X} - \widehat{\mathbf{X}}\|_2^2]$ is minimized.}

\subsection{Background}
Recently, significant research interest has been devoted to design and analysis of source coding, e.g. quantization, for CS, and a wide range of problems has been formulated. Existing work on this topic is mainly divided into three categories.

\begin{enumerate}

    \item The first category considers optimum quantizer design for quantization of CS measurements, where a CS reconstruction algorithm is held fixed at the decoder. Examples include \cite{09:Sun} and \cite{11:Kamilov}, where CS reconstruction algorithms are LASSO and message passing, respectively. Based on analysis-by-synthesis principle, we have recently developed a quantizer design method in \cite{13:Pasha_journal}, where any CS reconstruction algorithm can be used.

    \item The second category considers the design of a \textit{good} CS reconstruction algorithm, where the quantizer is held fixed. CS reconstruction from noisy measurements -- where the noise properties follow the effect of quantization -- falls in the category. Examples are \cite{10:Sinan,10:Zymnis,11:Dai,11:Jacques,12:Pasha1,12:Kamilov,08:Boufounos,12:Yan,13:Jacques,13:Plan}. To elaborate, let us consider \cite{11:Jacques} where CS measurements are uniformly quantized and a convex optimization-based CS reconstruction algorithm, called basis pursuit dequantizing (BPDQ), is developed  to suit the effect of uniform quantization. Further, the design of CS reconstruction algorithms and their performance bounds for reconstructing a sparse source from 1-bit quantized measurements have been investigated in \cite{08:Boufounos,12:Yan,13:Jacques,13:Plan}.

    \item Another line of previous work focuses on trade-offs between the quantization resources (e.g., quantization rate) and CS resources (e.g., number of measurements or complexity of CS reconstruction) \cite{06:Pai,08:Goyal,11:Dai,12:Laska}. For example, in \cite{12:Laska}, a trade-off between number of measurements and quantization rate was established by introducing the concept of two compression regimes as quantification of resources -- quantization compression regime and CS compression regime.

\end{enumerate}

We mention that all the above works are dedicated to pure source coding through quantization of CS measurements. To the best of our knowledge, there is limited work on JSCC of CS measurements using vector quantizer (VQ). In this regard, we had our previous effort in \cite{13:Pasha-icassp2}. The current paper is build upon the work of \cite{13:Pasha-icassp2}, and provides a comprehensive framework for developing optimum JSCC schemes to encode and transmit CS measurements (of a sparse source $\mathbf{X}$) over discrete memoryless channels, and to decode the sparse source so as to provide the reconstruction $\widehat{\mathbf{X}}$. The optimality is addressed by minimizing the MSE performance measure $\mathbb{E}[\|\mathbf{X} - \widehat{\mathbf{X}}\|_2^2]$. 

\subsection{Contributions}

We first consider the optimal design of VQ encoder-decoder pair for CS in the sense of minimizing the MSE. Here, we stress that we use the VQ in its \textit{generic} form. This is different from the design methods using uniform quantization \cite{11:Jacques} or 1-bit quantization of CS measurements  \cite{08:Boufounos,12:Yan,13:Jacques,13:Plan}. Our contributions include
\begin{itemize}
    \item Establishing (necessary) optimal encoding and decoding conditions for VQ.
    \item Providing a theoretical bound on the MSE performance.
    \item Developing a practical VQ encoder-decoder design through an iterative algorithm.
    \item Addressing the encoding complexity issue of VQ using a structured quantizer, namely low complexity multistage VQ (MSVQ), where we derive new encoder-decoder conditions for sub-optimal design of the MSVQ.
\end{itemize}
Our practical encoder-decoder designs consider Channel-Optimized VQ for CS, coined COVQ-CS, and Channel-Optimized MSVQ for CS, coined COMSVQ-CS. To demonstrate the strength of the proposed designs, we compare them with relevant quantizer design methods through different simulation studies. Particularly, we show that in noisy channel scenarios, the proposed COVQ-CS and COMSVQ-CS schemes provide better and more robust (against channel noise) performances compared to existing quantizers for CS followed by separate channel coding.

\subsection{Outline}
The rest of the paper is organized as follows. In \secref{sec:pre}, we introduce some preliminaries of CS. The optimal design and performance analysis of a joint source-channel VQ for CS are presented in \secref{sec:COVQ}. In \secref{subsec:train}, we propose a practical VQ encoder-decoder design algorithm. Further, in \secref{sec:COMSVQ}, we deal with complexity issue by proposing the design of computationally- and memory-efficient MSVQ for CS. The performance comparison of the proposed quantization schemes with other relevant methods are made in \secref{sec:numerical}, and conclusions are drawn in \secref{sec:conclusion}.

\subsection{Notations}

\textit{Notations:} Random variables (RV's) will be denoted by upper-case letters while their realizations (instants) will be denoted by the respective lower-case letters. Random vectors of dimension $n$ will be represented by boldface characters. 
We will denote a sequence of RV's $J_1,\ldots,J_N$ by $\mathbf{J}_1^N$; further, $\mathbf{J}_1^N = \mathbf{j}_1^N$ implies that $J_1=j_1,\ldots,J_N=j_N$. Matrices will be denoted by capital Greek letters, except that the square identity matrix of dimension $n$ is denoted by $\mathbf{I}_n$. The matrix operators determinant, trace, transpose and the maximum eigenvalue of a matrix are denoted by $\text{det}(\cdot)$, $\text{Tr}(\cdot)$, $(\cdot)^\top$, and $\lambda_{\max}(\cdot)$, respectively. Further, cardinality of a set is shown by $|\cdot|$.
We will use~$\mathbb{E}[\cdot]$ to denote the expectation operator. The $\ell_p$-norm ($p > 0$) of a vector $\mathbf{z}$ will be denoted by $\|\mathbf{z}\|_p = (\sum_{n=1}^N |z_n|^p)^{1/p}$. Also, $\|\mathbf{z}\|_0$ represents $\ell_0$-norm which is the number of non-zero coefficients in $\mathbf{z}$. 

\section{Preliminaries of CS} \label{sec:pre}
In CS, a random sparse vector (where most coefficients are likely zero) $\mathbf{X} \! \in  \! \mathbb{R}^N$ is linearly measured by a known sensing matrix $\mathbf{\Phi} \! \in \! \mathbb{R}^{M \times N}$ ($M \! < \! N$) resulting in an under-determined set of linear measurements (possibly) perturbed by noise
\begin{equation} \label{eq:measurement}
    \mathbf{Y = \Phi X + W},
\end{equation}
where $\mathbf{Y} \in \mathbb{R}^M$ and $\mathbf{W} \in \mathbb{R}^M$ denote the measurement and the additive measurement noise vectors, respectively. We assume that $\mathbf{X}$ is a $K$-sparse vector, i.e., it has at most $K$ ($K \leq M$) non-zero coefficients, where the location and magnitude of the non-zero components are drawn from known distributions. We also assume that the sparsity level $K$ is known in advance. We define the support set of the sparse vector $\mathbf{X} = [X_1,\ldots,X_N]^\top$ as $\mathcal{S} \triangleq \{n : X_n \neq 0 \} \subset \{1,\ldots,N\}$ with $|\mathcal{S}| = \|\mathbf{X}\|_0 \leq K$. Next, we define the mutual coherence notion which characterizes the merit of a sensing matrix $\mathbf{\Phi}$. The mutual coherence is defined as \cite{01:Donoho}
\begin{equation} \label{eq:mutual co}
    \mu \triangleq  \underset{i \neq j}{\max} \hspace{0.2cm} \frac{|\mathbf{\Phi}_i^\top \mathbf{\Phi}_j|}{\|\mathbf{\Phi}_i\|_2 \|\mathbf{\Phi}_j\|_2}, \hspace{0.2cm} 1 \leq i,j \leq N,
\end{equation}
where $\mathbf{\Phi}_i$ denotes the $i^{th}$ column of $\mathbf{\Phi}$. The mutual coherence $0 \leq \mu \leq 1$ formalizes the dependence between the columns of $\mathbf{\Phi}$, and can be calculated in polynomial-time complexity.

In order to reconstruct an unknown sparse source from a noisy under-sampled measurement vector, several reconstruction methods have been developed based on convex optimization methods, iterative greedy search algorithms and Bayesian estimation approaches. In this paper, through the design and analysis procedures, we adopt the Bayesian framework~\cite{07:Larsson,08:Ji,09:Elad,10:Protter,12:Kun} for reconstructing a sparse source from noisy and quantized measurements.

In the subsequent sections, we describe our proposed design methods for quantization by observing the CS measurement vector, and then develop theoretical results.

\section{Joint Source-Channel VQ for CS} \label{sec:COVQ}

In this section, we first introduce a general joint source-channel VQ system model for CS measurements in \secref{subsec:VQ}. We derive necessary conditions for optimality of encoder-decoder pair in \secref{sec:design COVQ}. Thereafter, we investigate the effects of optimal conditions in \secref{subsec:effetc}, and proceed to analysis of performance in \secref{subsec:bound COVQ-CS}. 

\begin{figure}
  \begin{center}
  \psfrag{x}[][][0.75]{$\mathbf{X}$}
  \psfrag{A}[][][1]{$\mathbf{\Phi}$}
  \psfrag{y}[][][0.75]{$\mathbf{Y}$}
  \psfrag{w}[][][0.75]{$\mathbf{W}$}
  \psfrag{Q}[][][0.85]{$\textsf{E}$}
  \psfrag{C}[][][1]{$\mathcal{C}$}
  \psfrag{i}[][][0.75]{$I$}
  \psfrag{j}[][][0.75]{$J$}
  \psfrag{DMC}[][][0.8]{$P(j|i)$}
  \psfrag{Channel}[][][0.75]{Channel}
  \psfrag{quant}[][][0.75]{Quantizer}
  \psfrag{Enc}[][][0.75]{encoder}
  \psfrag{CS Enc}[][][0.75]{CS sensing}
  \psfrag{Dec}[][][0.75]{Decoder}
  \psfrag{Dec}[][][0.85]{Decoder}
  \psfrag{D}[][][0.85]{$\textsf{D}$}
  \psfrag{x_h}[][][0.75]{$\widehat{\mathbf{X}}$}
  \includegraphics[width=9cm]{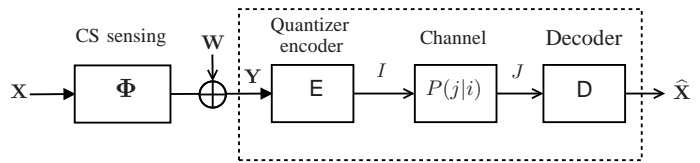}\\
  \caption{Studied system model for joint source-channel vector quantization of CS measurements. The goal is to design encoder and decoder mappings (illustrated in dashed box) with respect to minimizing $\mathbb{E}[\|\mathbf{X} - \widehat{\mathbf{X}}\|_2^2]$ while the CS sensing matrix $\mathbf{\Phi}$ and channel $P(j|i)$ are known in advance.}\label{fig:diagram_VQ}
  \end{center}
\end{figure}

\subsection{General System Description and Performance Criterion} \label{subsec:VQ}

Consider the general system model, shown in \figref{fig:diagram_VQ}, for transmitting CS measurements and reconstructing a sparse source. Let the total bit budget allocated for encoding (quantization) be fixed at $R$ bits per dimension of the source vector. Given the noisy measurement vector $\mathbf{Y}$, a VQ encoder is defined by a mapping $\textsf{E}:\! \mathbb{R}^M \rightarrow \! \mathcal{I}$, where $\mathcal{I}$ is a finite index set defined as $\mathcal{I} \!\triangleq \{0,1,\ldots,2^{R}-1\}$ with $|\mathcal{I}| \! \triangleq \mathfrak{R}=2^{R}$. Denoting the quantized index by $I$, the encoder works according to $\mathbf{Y} \! \in \mathcal{R}_i \Rightarrow I \!= i$, where the sets $\{\mathcal{R}_i\}_{i=0}^{\mathfrak{R}\!-1}$ are encoder regions and $\bigcup_{i=0}^{\mathfrak{R}\!-1} \mathcal{R}_{i} =\! \mathbb{R}^M$ such that when $\mathbf{Y} \in \mathcal{R}_i$ the encoder outputs the index $\textsf{E}(\mathbf{Y}) =\! i \in \mathcal{I}$. Note that given an index $i$, the set $\mathcal{R}_{i}$ is not necessarily a connected set (due to non-linear CS reconstruction) in the space $\mathbb{R}^M$. Also, $\mathcal{R}_{i}$ might be an empty set (due to channel noise, see e.g. \cite{91:Farvardin}).

Next, we consider a memoryless channel consisting of discrete input and output alphabets which is referred to as discrete memoryless channel (DMC). 
In our problem setup, the DMC accepts the encoded index $i$ and outputs a noisy symbol $j \!\in\! \mathcal{I}$. The channel is defined by a random mapping characterized by transition probabilities
\begin{equation} \label{eq:channel trans}
P(j | i) \triangleq \textrm{Pr}(J = j | I = i), \hspace{0.15cm} i,j \in \mathcal{I},
\end{equation}
which indicates the probability that index $j$ is received given that the input index to the channel was $i$. We assume that the transmitted index $i$ and the received index $j$ share the same index set $\mathcal{I}$, and the channel transition probabilities \eqref{eq:channel trans} are known in advance. We denote the capacity of a given channel by $C$ bits/channel use. Given the received index $j$, a decoder is characterized by a mapping $\textsf{D}\!: \mathcal{I} \!\rightarrow \! \mathcal{C}$ where $\mathcal{C}$ is a finite discrete \textit{codebook} set containing all reproduction \textit{codevectors} $\{\mathbf{c}_{j} \!\in \!\mathbb{R}^N\}_{j=0}^{\mathfrak{R}\!-\!1}$. The decoder's functionality is described by a look-up table; $J \!=\! j \Rightarrow \widehat{\mathbf{X}} \!=\! \mathbf{c}_j$ such that when the received index from the channel is $j$, the decoder outputs  $\textsf{D}(j)\!=\!\mathbf{c}_j \!\in\! \mathcal{C}$.

Next, we state how we quantify the performance of \figref{fig:diagram_VQ} and our design goal.
It is important to design an encoder-decoder pair in order to minimize a distortion measure which reflects the requirements of the receiving-end user. Therefore, we quantify the source reconstruction distortion of our studied system by the end-to-end MSE defined as
\begin{equation} \label{eq:e2e dist}
    D \triangleq \mathbb{E}[\|\mathbf{X - \widehat{X}}\|_2^2],
\end{equation}
where the expectation is taken with respect to the distributions on the sparse source $\mathbf{X}$ (which, itself, depends on the distribution of non-zero coefficients in $\mathbf{X}$ as well as their random placements (sparsity pattern)), the noise $\mathbf{W}$ and the randomness in the channel. We mention that the end-to-end MSE depends on \textit{CS reconstruction error}, \textit{quantization error} as well as \textit{channel noise}. While the CS sensing matrix $\mathbf{\Phi}$ is given, our concern is to design an encoder-decoder pair robust against all these three kinds of error.
\subsection{Optimality Conditions for VQ Encoder and Decoder} \label{sec:design COVQ}
We consider an optimization technique for the system illustrated in \figref{fig:diagram_VQ} in order to determine encoder and decoder mappings \textsf{E} and \textsf{D}, respectively, in the presence of channel noise. More precisely, the aim of the VQ design is to find
\begin{itemize}
    \item MSE-minimizing encoder regions $\{\mathcal{R}_i\}_{i=0}^{\mathfrak{R}-1}$ and
    \item MSE-minimizing decoder codebook $\mathcal{C} = \{\mathbf{c}_j\}_{j=0}^{\mathfrak{R}-1}$.
\end{itemize}
We note that the optimal joint design of encoder and decoder cannot be implemented since the resulting optimization is analytically intractable. To address this issue, in \secref{subsec:opt enc}, we show how the encoding index $i \in \mathcal{I}$ (or equivalently encoder region $\mathcal{R}_i$) can be chosen to minimize the MSE for a given codebook $\mathcal{C}=\{\mathbf{c}_j\}_{j=0}^{\mathfrak{R}-1}$. Then, in \secref{subsec:opt dec}, we derive an expression for the optimal decoder codebook $\mathcal{C}$ for given encoder regions $\{\mathcal{R}_i\}_{i=0}^{\mathfrak{R}-1}$.

\subsubsection{Optimal Encoder} \label{subsec:opt enc}

First, let us introduce the \textit{minimum mean-square error} (MMSE) estimator of the source given the observed measurements \eqref{eq:measurement} which is (see \cite[Chapter 11]{93:Kay})
\begin{equation} \label{eq:MMSE sparse}
    \widetilde{\mathbf{x}}(\mathbf{y}) \triangleq  \mathbb{E} [\mathbf{X} | \mathbf{Y} = \mathbf{y}] \in \mathbb{R}^N.
\end{equation}

Now, assume that the decoder codebook $\mathcal{C}=\{\mathbf{c}_j\}_{j=0}^{\mathfrak{R}-1}$ is known and fixed. We focus on how the encoding index $i$ should be chosen to minimize the MSE given the observed noisy CS measurement vector $\mathbf{y}$. We rewrite the MSE as
\begin{equation} \label{eq:rewrite MSE}
\begin{aligned}
    D &\triangleq \mathbb{E}[\|\mathbf{X - \widehat{X}}\|_2^2] = \mathbb{E}[\|\mathbf{X} - \mathbf{c}_J\|_2^2]& \\
    &\stackrel{(a)}{=} \! \int_{\mathbf{y}} \sum_{i \in \mathcal{I}} \textrm{Pr} \{I\!=\!i | \mathbf{Y\!=\!y}\} \mathbb{E} \left[\|\mathbf{X} \!-\! \mathbf{c}_J \|_2^2 | \mathbf{Y\!=\!y} , I\!=\!i\right] f (\mathbf{y}) d\mathbf{y}& \\
    &\stackrel{(b)}{=} \sum_{i \in \mathcal{I}} \int_{\mathbf{y}  \in \mathcal{R}_i} \bigg\{ \mathbb{E} \left[\|\mathbf{X} - \mathbf{c}_J \|_2^2 | \mathbf{Y=y} , I=i\right] \bigg\} f (\mathbf{y}) d\mathbf{y} ,&
\end{aligned}
\end{equation}
where $(a)$ follows from marginalization of the MSE over $\mathbf{Y}$ and $I$. Further, $f(\mathbf{y})$ is the $M$-fold probability density function (pdf) of the measurement vector.
Also, $(b)$ follows by interchanging the integral and the summation and the fact that $\textrm{Pr} \{I=i | \mathbf{Y=y}\} = 1$, $\forall \mathbf{y} \in \mathcal{R}_i$, and otherwise the probability is zero. Now, since $f(\mathbf{y})$ is always non-negative, the MSE-minimizing points in $\mathbb{R}^M$ that shall be assigned to the encoder region $\mathcal{R}_i$ are those that minimize the term within the braces in the last expression of \eqref{eq:rewrite MSE}. Then, the MSE-minimizing encoding index, denoted by $i^\star \in \mathcal{I}$, is given by
\begin{equation} \label{eq:proof enc}
\begin{aligned}
    i^\star &= \textrm{arg }\underset{i \in \mathcal{I}}{\textrm{min }} \mathbb{E} \left[\|\mathbf{X} - \mathbf{c}_J \|_2^2 | \mathbf{Y=y} , I=i\right]& \\
    &\stackrel{(a)}{=} \textrm{arg }\underset{i \in \mathcal{I}}{\textrm{min }} \big\{ \mathbb{E}[\|\mathbf{c}_J\|_2^2  | \mathbf{Y\!=\!y}, I\!=\!i] \!-\! 2\mathbb{E}[\mathbf{X}^\top \mathbf{c}_J  | \mathbf{Y\!=\!y}, I\!=\!i] \big\}& \\
    &\stackrel{(b)}{=} \textrm{arg }\underset{i \in \mathcal{I}}{\textrm{min }} \big\{ \mathbb{E}[\|\mathbf{c}_J\|_2^2 \big | I=i] \!-\! 2\mathbb{E}[\mathbf{X}^\top \big | \mathbf{Y=y}] \mathbb{E}[\mathbf{c}_J \big | I=i] \big\},&
\end{aligned}
\end{equation}
where $(a)$ follows from the fact that $\mathbf{X}$ is independent of $I$, conditioned on $\mathbf{Y}$; hence, $\mathbb{E} \left[\|\mathbf{X}\|_2^2 | \mathbf{Y\!=\!y}, I\!=\!i \right] = \mathbb{E} \left[\|\mathbf{X}\|_2^2 | \mathbf{Y\!=\!y}\right]$ which is pulled out of the optimization. $(b)$ follows from the fact that $\mathbf{c}_J$ is independent of $\mathbf{Y}$, conditioned on $I$, and from the Markov chain $\mathbf{X} \rightarrow \mathbf{Y} \rightarrow I \rightarrow \mathbf{c}_J$. Next, note that introducing channel transition probabilities $P(j|i)$ in \eqref{eq:channel trans} and the MMSE estimator $\widetilde{\mathbf{x}}(\mathbf{y})$ in \eqref{eq:MMSE sparse}, the last equality in \eqref{eq:proof enc} can be expressed as
\begin{equation} \label{eq:final enc}
    i^\star = \textrm{arg }\underset{i \in \mathcal{I}}{\textrm{min}} \left\{ \sum_{j=0}^{\mathfrak{R}-1} P(j|i) \left\| \mathbf{c}_j \right\|_2^2 - 2 \widetilde{\mathbf{x}}(\mathbf{y})^\top \sum_{j=0}^{\mathfrak{R}-1} P(j|i) \mathbf{c}_j \right\}.
\end{equation}
Equivalently, the optimized encoding regions are obtained by 
\begin{equation} \label{eq:regions}
\begin{aligned}
    \mathcal{R}_{i}^\star  =  &\left\{\mathbf{y} \in \mathbb{R}^M  : \sum_{j=0}^{\mathfrak{R}-1} \left[P(j|i) - P(j|i')\right] \left\| \mathbf{c}_j \right\|_2^2 \leq \right.& \\
    & \hspace{0.5cm} \left. 2 \widetilde{\mathbf{x}}(\mathbf{y})^\top \sum_{j=0}^{\mathfrak{R}-1} \left[P(j|i) - P(j|i') \right] \mathbf{c}_j , i \neq i' \in \mathcal{I} \right\}.&
\end{aligned}
\end{equation}

\subsubsection{Optimal Decoder} \label{subsec:opt dec}

Applying the MSE criterion, it is straightforward to show that the codevectors which minimize $D$ in \eqref{eq:e2e dist} for a fixed encoder are obtained by letting $\mathbf{c}_j$ represent the MMSE estimator of the vector $\mathbf{X}$ based on the received index $j$ from the channel, that is
\begin{equation} \label{eq:opt dec}
    \mathbf{c}_j^\star =  \mathbb{E}[\mathbf{X} | J=j] , \hspace{0.15cm} j \in \mathcal{I}.
\end{equation}
Now, using the Bayes' rule, the expression for $\mathbf{c}_j^\star$ can be rewritten as
\begin{equation} \label{eq:opt dec final}
\begin{aligned}
    \mathbf{c}_j^\star &=  \mathbb{E}[\mathbf{X} | J=j]&\\
    &= \sum_{i} P(i | j) \mathbb{E}[\mathbf{X} | J=j,I=i]&\\
    &\stackrel{(a)}{=}  \frac{\sum_i P(j|i) P(i) \int_{\mathbf{y}}\mathbb{E}[\mathbf{X}|\mathbf{Y\!=\!y}] f(\mathbf{y}|i) d\mathbf{y}}{\sum_i P(j|i) P(i)}& \\
    &\stackrel{(b)}{=} \frac{\sum_i P(j|i) \int_{\mathcal{R}_i} \widetilde{\mathbf{x}}(\mathbf{y}) f(\mathbf{y}) d\mathbf{y} }{\sum_i P(j|i) \int_{\mathcal{R}_i} f(\mathbf{y}) d\mathbf{y}},&
\end{aligned}
\end{equation}
where $(a)$ follows from marginalization over $\mathbf{Y}$ and the Markov chain $\mathbf{X} \rightarrow \mathbf{Y} \rightarrow I$. Moreover, $f(\mathbf{y}|i)$ is the conditional pdf of $\mathbf{Y}$ given that $\mathbf{Y} \in \mathcal{R}_i$. Also, $(b)$ follows by using \eqref{eq:MMSE sparse} and by the fact that $f(\mathbf{y}|i)=0$, $\forall \mathbf{y} \notin \mathcal{R}_i$.

The optimal conditions in \eqref{eq:final enc} and \eqref{eq:opt dec final} can be used in an \textit{alternate-iterate} procedure to design a practical encoder-decoder pair for vector quantization of CS measurements. The resulting algorithm will be presented later in \secref{subsec:train}.

\subsection{Insights Through Analyzing the Optimal Conditions} \label{subsec:effetc}
Here, we provide insights into the necessary optimal conditions \eqref{eq:final enc} and \eqref{eq:opt dec}. Note that the encoding condition \eqref{eq:final enc} implies that the sparse source is first MMSE-wise reconstructed from CS measurements at the encoder, and then quantized to an appropriate index. Hence, it suggests that the system shown in \figref{fig:diagram_VQ} may be translated to the equivalent system shown in \figref{fig:decomposed VQ}.
\begin{figure}[!ht]
    \begin{center}
        \psfrag{x}[][][0.7]{$\mathbf{X}$}
        \psfrag{A}[][][0.75]{$\mathbf{\Phi}$}
        \psfrag{y}[][][0.7]{$\mathbf{Y}$}
        \psfrag{w}[][][0.7]{$\mathbf{W}$}
        \psfrag{Q}[][][0.7]{$\textsf{E}$}
        \psfrag{i}[][][0.7]{$I$}
        \psfrag{j}[][][0.7]{$J$}
        \psfrag{DMC}[][][0.7]{$P(j|i)$}
        \psfrag{Channel}[][][0.7]{Channel}
        \psfrag{quant}[][][0.7]{Quantizer}
        \psfrag{Enc}[][][0.7]{encoder}
        \psfrag{CS Enc}[][][0.7]{CS sensing}
        \psfrag{CS D}[][][0.7]{CS decoder}
        \psfrag{A-1}[][][0.7]{\textsf{R}}
        \psfrag{MMSE}[][][0.5]{(MMSE)}
        \psfrag{y_h}[][][0.7]{$\widetilde{\mathbf{X}}$}
        \psfrag{Dec}[][][0.7]{decoder}
        \psfrag{D}[][][0.7]{$\textsf{D}$}
        \psfrag{x_h}[][][0.7]{$\widehat{\mathbf{X}}$}
        \includegraphics[width=\columnwidth,height=1.8cm]{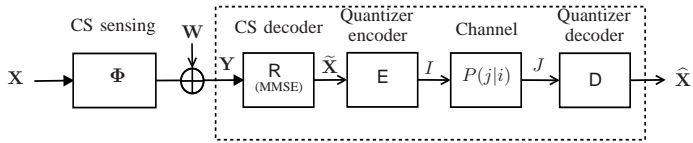}\\
        \caption{ Equivalent block diagram of a system with CS reconstruction $\textsf{R}$ (MMSE reconstruction) at the encoder side. Necessary optimal conditions for encoder-decoder pair of this system are equivalent to those of the original system model shown in \figref{fig:diagram_VQ}.}
       \label{fig:decomposed VQ}
    \end{center}
\end{figure}

Let us first denote the MMSE estimator as the RV $\widetilde{\mathbf{X}}(\mathbf{Y}) \triangleq \mathbb{E}[\mathbf{X}|\mathbf{Y}]$, then we rewrite the end-to-end distortion $D$ as
\begin{equation} \label{eq:MSE COVQ}
\begin{aligned}
    D &=\mathbb{E}[\|\mathbf{X} -  \widetilde{\mathbf{X}}(\mathbf{Y}) + \widetilde{\mathbf{X}}(\mathbf{Y}) - \widehat{\mathbf{X}}\|_2^2]&\\
    &= \mathbb{E}[\|\mathbf{X} -  \widetilde{\mathbf{X}}(\mathbf{Y})\|_2^2] + \mathbb{E}[\|\widetilde{\mathbf{X}}(\mathbf{Y}) - \widehat{\mathbf{X}}\|_2^2],&
\end{aligned}
\end{equation}
where the second equality can be proved by showing that the estimation error of the source $\mathbf{X} - \widetilde{\mathbf{X}}(\mathbf{Y})$ and the quantized transmission error $\widetilde{\mathbf{X}}(\mathbf{Y}) - \widehat{\mathbf{X}}$ are uncorrelated. This holds from the definition of $\widetilde{\mathbf{X}}(\mathbf{Y})$ and the long Markov property $\mathbf{X} \rightarrow \mathbf{Y} \rightarrow I \rightarrow J \rightarrow \widehat{\mathbf{X}}$ due to the assumption of deterministic mappings $\textsf{E}$ and $\textsf{D}$ and memoryless channel.
\begin{rem}
    Following \eqref{eq:MSE COVQ}, let us denote by $D_{cs} \triangleq \mathbb{E}[\|\mathbf{X} -  \widetilde{\mathbf{X}}(\mathbf{Y})\|_2^2]$ the CS reconstruction distortion, and by $D_{q} \triangleq \mathbb{E}[\|\widetilde{\mathbf{X}}(\mathbf{Y}) - \widehat{\mathbf{X}} \|_2^2]$ the quantized transmission distortion. Then, the decomposition \eqref{eq:MSE COVQ}  indicates that the end-to-end source distortion $D$, without loss of optimality, is equivalent to $D = D_{cs} + D_q$.
\end{rem}

Interestingly, it can be also seen from \eqref{eq:MSE COVQ} that $D_{cs}$ does not depend on quantization and channel aspects. Hence, to find optimal encoding indexes (given fixed codevectors) and optimal codevectors (given fixed encoding regions) with respect to the end-to-end distortion $D$, it suffices to find them with respect to minimizing $D_q$. It can be proved that the necessary conditions for optimality (with respect to $D_q$) of the encoder-decoder pair derived for the system of \figref{fig:decomposed VQ} coincide with the ones developed for the system of \figref{fig:diagram_VQ}, i.e., \eqref{eq:final enc} and \eqref{eq:opt dec final}. The proof of this claim is as follows. Similar to the steps taken in \eqref{eq:rewrite MSE}, the $D_q$--minimizing encoding index $i^\star \in \mathcal{I}$ is given by
\begin{equation*} \label{eq:COVQCS_proof enc_equi}
\begin{aligned}
    i^\star &= \textrm{arg }\underset{i \in \mathcal{I}}{\textrm{min }} \mathbb{E} \left[\|\widetilde{\mathbf{X}}(\mathbf{Y}) - \mathbf{c}_J \|_2^2 | \mathbf{Y=y} , I=i\right]& \\
    &= \textrm{arg }\underset{i \in \mathcal{I}}{\textrm{min }} \left\{ \mathbb{E}[\|\mathbf{c}_J\|_2^2 \big | I=i] \!-\! 2\widetilde{\mathbf{x}}(\mathbf{y})^\top  \mathbb{E}[\mathbf{c}_J \big | I=i] \right\},& \\
    &=\textrm{arg }\underset{i \in \mathcal{I}}{\textrm{min}} \left\{ \sum_{j=0}^{\mathfrak{R}-1} P(j|i) \left\| \mathbf{c}_j \right\|_2^2 - 2 \widetilde{\mathbf{x}}(\mathbf{y})^\top \sum_{j=0}^{\mathfrak{R}-1} P(j|i) \mathbf{c}_j \right\}.&
\end{aligned}
\end{equation*}
Further, the $D_q$--minimizing decoder $\mathbf{c}_j^\star$ is obtained by
\begin{equation*} \label{eq:COVQCS_equi dec}
\begin{aligned}
    \mathbf{c}_j^\star &= \mathbb{E}[\widetilde{\mathbf{X}}^\star(\mathbf{Y}) | J = j]& \\
    &\stackrel{(a)}{=} \int \mathbb{E}[\mathbf{X}|J = j, \mathbf{Y=y}] p(\mathbf{y}|j) d\mathbf{y}&\\
    &= \mathbb{E} [\mathbf{X} | J = j],&
\end{aligned}
\end{equation*}
where $(a)$ follows from the Markov property $\widetilde{\mathbf{X}}(\mathbf{Y}) \rightarrow \mathbf{Y} \rightarrow J$. Now, we provide the following remark.
\begin{rem} \label{rem:equivalence}
    The general system of \figref{fig:diagram_VQ} and \figref{fig:decomposed VQ} are equivalent considering end-to-end MSE criterion, fixed sensing matrix and channel transition probabilities.
\end{rem}

Before proceeding to the analysis of the MSE using the developed equivalence property, we provide a comparative study between our proposed design scheme with related methods in the literature which follow the building block structure shown in \figref{fig:subopt}. Under this system model, for a fixed CS reconstruction algorithm (or, a fixed quantizer encoder-decoder pair), a quantizer encoder-decoder pair (or, CS reconstruction algorithm) is designed in order to satisfy a certain performance criterion, e.g. minimizing end-to-end distortion, quantization distortion or $\ell_1$--norm of reconstruction vector. Some examples of system models following \figref{fig:subopt} include \cite{09:Sun,11:Jacques,12:Kamilov,13:Pasha_journal} (assuming a noiseless channel) and the conventional \textit{nearest-neighbor coding} of CS measurements. In general, according to this system model, quantizer decoder $\textsf{D}$ outputs the vector $\widehat{\mathbf{Y}} \in \mathbb{R}^M$ after receiving channel output. Finally, a given CS reconstruction decoder $\textsf{R}:\mathbb{R}^M \!\rightarrow \!\mathbb{R}^N$ takes $\widehat{\mathbf{Y}}$ and makes an estimate of the sparse source.
\begin{figure}[!ht]
    \begin{center}
        \psfrag{x}[][][0.7]{$\mathbf{X}$}
        \psfrag{A}[][][0.75]{$\mathbf{\Phi}$}
        \psfrag{y}[][][0.7]{$\mathbf{Y}$}
        \psfrag{w}[][][0.7]{$\mathbf{W}$}
        \psfrag{Q}[][][0.7]{$\textsf{E}$}
        \psfrag{i}[][][0.7]{$U$}
        \psfrag{j}[][][0.7]{$V$}
        \psfrag{DMC}[][][0.7]{$P(v|u)$}
        \psfrag{Channel}[][][0.7]{Channel}
        \psfrag{quant}[][][0.7]{Quantizer}
        \psfrag{Enc}[][][0.7]{encoder}
        \psfrag{CS Enc}[][][0.7]{CS sensing}
        \psfrag{CS D}[][][0.7]{CS decoder}
        \psfrag{A-1}[][][0.7]{\textsf{R}}
        \psfrag{y_h}[][][0.7]{$\widehat{\mathbf{Y}}$}
        \psfrag{Dec}[][][0.7]{decoder}
        \psfrag{D}[][][0.7]{$\textsf{D}$}
        \psfrag{x_h}[][][0.7]{$\widehat{\mathbf{X}}$}
        \includegraphics[width=\columnwidth,height=1.8cm]{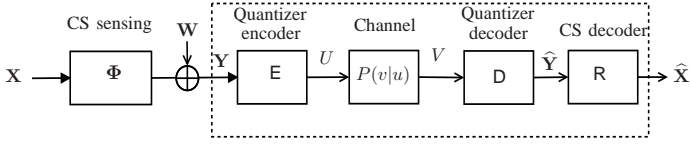}\\
        \caption{ Block diagram of a system with CS reconstruction at the decoder side. The aim is to design encoder--decoder mappings or CS decoder (illustrated in dashed box) with respect to optimizing a performance criterion while the CS sensing matrix $\mathbf{\Phi}$ and channel $P(v|u)$ are known in advance.}
       \label{fig:subopt}
    \end{center}
\end{figure}

Following \figref{fig:decomposed VQ} (as the equivalent system model of \figref{fig:diagram_VQ}), we note that it is structurally different from the system model of \figref{fig:subopt} in the location of the CS reconstruction, either at the transmitter side or at the receiver side. In the former system, an encoder reconstructs the source from CS measurements, whereas the latter system puts all CS reconstruction complexity at the decoder.

\subsection{Analysis of MSE} \label{subsec:bound COVQ-CS}
In this section, we provide an analysis into the impact of CS reconstruction distortion, quantization error and channel noise on the end-to-end MSE by deriving a lower-bound.

\begin{proposition} \label{theo2}
Consider the linear CS model \eqref{eq:measurement} with an exact $K$-sparse source $\mathbf{X} \in \mathbb{R}^N$ under the following assumptions:
\begin{enumerate}[i.]
    \item The magnitude of $K$ non-zero coefficients in $\mathbf{X}$ are drawn according to the i.i.d. standard Gaussian distribution.
    \item The $K$ elements of the support set are uniformly drawn from all ${N \choose K}$ possibilities.
    \item The measurement noise is drawn as $\mathbf{W} \sim \mathcal{N}(\mathbf{0}, \sigma_w^2 \mathbf{I}_M)$ uncorrelated with the measurements, where $\sigma_w^2 \neq 0$.
\end{enumerate}
Further, assume a sensing matrix $\mathbf{\Phi}$ with mutual coherence $\mu$. Let the total quantization rate be $R$ bits/vector, and the channel be characterized by capacity $C$ bits/channel use, then the end-to-end MSE of the system of \figref{fig:diagram_VQ} asymptotically (in quantization rate and dimension) is lower-bounded as
\begin{equation} \label{eq:lower-bound SSC}
\begin{aligned}
        D \geq Kc_1 + c_1 c_2  2^{-2C\left(\frac{R - \log_2 {N \choose K} }{K} \right)} ,
\end{aligned}
\end{equation}
where
    $c_1 \!=\! \frac{\sigma_w^2}{1 + \sigma_w^2 + (K+1)\mu}$,
and
    $c_2 \!=\! 2\left(\frac{K}{2} \Gamma\left(\frac{K}{2}\right) \right)^{\frac{2}{K}} \left(\frac{K+2}{K}\right)^{\frac{K}{2}}$, in which $\Gamma(\cdot)$ denotes the Gamma function.
\end{proposition}

\begin{proof}
	The proof can be found in the Appendix.
\end{proof}

\begin{rem} \label{rem:MSE floor}
    Each component of the lower-bound \eqref{eq:lower-bound SSC} is intuitive. The first term is the contribution of the CS reconstruction distortion, and the second term reflects the distortion due to the vector quantized transmission. When the CS measurements are noisy, it can be verified that as $R$ increases, the end-to-end MSE attains an error floor. This result can be also inferred from \eqref{eq:MSE COVQ}: as quantization rate increases, $D_q$ decays (asymptotically) exponentially, however, $D_{cs}$ is constant irrespective of rate. Hence, as $R \rightarrow \infty$, the value that the MSE converges to is $D_{cs} = \mathbb{E}[\|\mathbf{X} - \widetilde{\mathbf{X}}\|_2^2]$.
\end{rem}

It should be noted when CS measurements are noiseless ($\sigma_w^2 = 0$), the lower-bound \eqref{eq:lower-bound SSC} becomes trivial. In this case, a simple asymptotic lower-bound for the system of \figref{fig:diagram_VQ}, under the assumptions of \proref{theo2}, can be obtained as
\begin{equation} \label{eq:simple lb}
    \begin{aligned}
        D \geq c_2  2^{-2C\left(\frac{R - \log_2 {N \choose K} }{K} \right)} ,
    \end{aligned}
\end{equation}
where the constant $c_2$ is the same dimensionality-dependent constant in \eqref{eq:lower-bound SSC}.

The lower-bound \eqref{eq:simple lb} (also known as adaptive bound in \cite{06:Pai,08:Goyal} in the noiseless channel case) can be proved assuming that the support set of $\mathbf{X} \in \mathbb{R}^N$ is \textit{a priori} known. Therefore, one can transmit the known support set using $\log_2 {N \choose K}$ bits, and the Gaussian coefficients within the support set can be quantized via $R - \log_2{N \choose K}$ bits.  Under noiseless channel condition ($C=1$), the right hand side in \eqref{eq:simple lb} is shown to achieve the distortion rate function of a $K$-sparse source vector with Gaussian non-zero coefficients and a support set uniformly drawn from ${N \choose K}$ possibilities \cite{12:Weidmann}. Then, the separate source-channel coding theorem \cite[Chapter 7]{06:Cover} can be applied to find the optimum performance theoretically attainable (OPTA) by introducing channel capacity $C$.

\begin{rem}
The lower-bound in \eqref{eq:simple lb} shows that the end-to-end MSE can at most decay exponentially (in quantization rate $R$) with exponent $-\frac{6C}{K}$ dB/bit. Since the sparsity ratio $\frac{K}{N}<1$, the decaying exponent can be far steeper than $-\frac{6C}{N}$ dB/bit for a Gaussian non-sparse source vector of dimension $N$.
\end{rem}

The following toy example offers some insights into the tightness of the lower-bound \eqref{eq:simple lb}.

\begin{ex} \label{ex:lb}
    Using a simple example, we show how tight the lower-bound \eqref{eq:simple lb} is with respect to our proposed design. In \figref{fig:LB_rate}, we compare simulation results with the lower-bound in some region where $\widetilde{\mathbf{X}}(\mathbf{Y}) \rightarrow \mathbf{X}$.\footnote[1]{This scenario can be realized in an event where $\sigma_w^2 = 0$ and number of measurements is such that the CS reconstruction is perfect.} Following this best-case scenario, we generate $2 \times 10^5$ realizations of $\mathbf{X} \in \mathbb{R}^2$ with sparsity level $K=1$, where the non-zero coefficient is a standard Gaussian RV, and its location is drawn uniformly at random over $\{1,2\}$. Then, we use the necessary optimal conditions \eqref{eq:final enc} and \eqref{eq:opt dec} iteratively (as will be shown later in \algref{alg:Lloyd}). Considering a binary symmetric channel (BSC) with bit cross-over probability $\epsilon$ and capacity $C$ bits/channel use (see \eqref{eq:capacity BSC}), we plot MSE, $D=\mathbb{E}[\|\mathbf{X} - \widehat{\mathbf{X}}\|_2^2]$ versus quantization rate $R$ for $\epsilon = 0$ (noiseless channel) and $\epsilon = 0.02$ (noisy channel) in \figref{fig:LB_rate}. It can be observed that at $\epsilon=0$, the bound (dashed line) is tight. As would be expected, degrading channel condition to $\epsilon = 0.02$ reduces the performance. At $\epsilon=0.02$, the gap between the simulation result (solid line marked by `o') and its corresponding lower-bound (dotted line) increases. Note that in the noisy channel case, the lower-bound is based upon the asymptotic assumption of infinite source and channel code lengths (used in the OPTA). Therefore, the lower-bound is not tight at $\epsilon=0.02$ for low dimensions. 

\begin{figure}
  \begin{center}
  \includegraphics[width=\columnwidth,height=7.5cm]{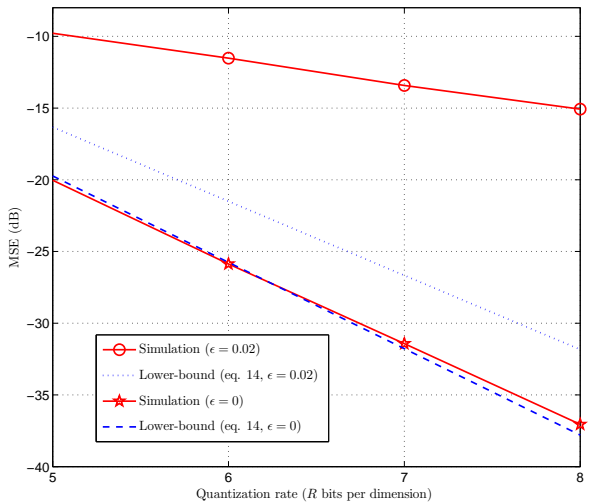}\\
  \caption{Comparison of the lower-bound in \eqref{eq:simple lb} and simulation results for simulation of a 1-sparse source $\mathbf{X} \in \mathbb{R}^2$ in a region where the locally reconstructed source $\widetilde{\mathbf{X}}$ can be perfectly recovered from noiseless measurements.}
  \label{fig:LB_rate}
  \end{center}
\end{figure}

\end{ex}

\section{Practical Quantizer Design} \label{subsec:train}
In this section, we first develop a practical VQ encoder-decoder design algorithm, referred to as channel-optimized VQ for CS (COVQ-CS) using the necessary optimal conditions \eqref{eq:final enc} and \eqref{eq:opt dec final}. Then, we provide a practical comparison between our proposed algorithm and a conventional quantizer design algorithm. We finalize this section by analyzing encoding and decoding computational complexity.

\subsection{Training Algorithm for Practical Design}

The results presented in \secref{subsec:opt enc} and \secref{subsec:opt dec} can be utilized to formulate an \textit{iterate-alternate} training algorithm for the problem of interest. Similar to the \textit{generalized Lloyd} algorithm for noisy channels \cite{90:Farvardin}, we propose a VQ training method for the design problem in this paper which is summarized in \algref{alg:Lloyd}. The following remarks can be considered for implementing \algref{alg:Lloyd}:
\begin{itemize}
    \item In step (1), besides the channel transition probabilities $P(j|i)$, we assume that the statistics of the sparse source vector are given for training. 
    \item In general, it is not easy to derive closed-form solutions for the optimal decoding condition \eqref{eq:opt dec final}, for example, due to difficulties in calculating the integrals even if the pdf $f(\mathbf{y})$ is known. In practice, we calculate the codevector $\mathbf{c}_j$ ($j \in \mathcal{I}$) in \eqref{eq:opt dec} using the Monte-Carlo method. To implement this computationally-efficient procedure, we first generate a set of finite \textit{training vectors} $\mathbf{X}$, and then sample-average over those vectors that have led to the index $J=j$.
    \item To address the issue of encountering empty regions, we, in each iteration of the algorithm, pick the codevector whose index has been sent the most number of times, denoted by
        $\mathbf{c}_j^{\max}$. Then, a codevector associated with the index that has not been sent is calculated as $\mathbf{c}_j^{\max} + \delta \mathbf{c}_j^{\max}$, where $\delta>0$ is sufficiently small. Using this technique (which is also known as splitting method in the initialization phase of the LBG algorithm \cite{80:LBG}), we efficiently re-include those encoding indexes that have never been selected due to the limited number of generated samples. This will lead to a design that efficiently uses all degrees of freedom.

    \item The performance of the COVQ-CS is sensitive to initializations in order for the algorithm to converge to a smaller value of the distortion $D$. Therefore, in step (3), when the channel is noiseless, the codevectors are initialized using the splitting procedure of the so-called LBG design algorithm. Then, the final optimized codevectors are chosen for initialization of \algref{alg:Lloyd} in the noisy channel case. Furthermore, convergence in step (7) may be checked by tracking the MSE, and terminate the iterations when the relative improvement is small enough. By construction and ignoring issues such as numerical precision, the iterative design in \algref{alg:Lloyd} always converges to a local optimum since when the criteria in steps (5) and (6) of the algorithm are invoked, the performance can only leave unchanged or improved, given the updated indexes and codevectors. This is a common rationale behind the proof of convergence for such iterative algorithms (see e.g. \cite[Lemma 11.3.1]{91:Gersho}). However, nothing can be generally guaranteed about the global optimality of this algorithm.
\end{itemize}

\begin{algorithm}
\caption{ COVQ-CS: Practical training algorithm}\label{alg:Lloyd}
\begin{algorithmic}[1]
\STATE{\textbf{input:} measurement vector: $\mathbf{y}$, channel probabilities: $P(j|i)$, bit budget: $R$ bits/vector.}
\STATE{\textbf{compute:} $\widetilde{\mathbf{x}}(\mathbf{y})$ in \eqref{eq:MMSE sparse}.}
\STATE{\textbf{initialize: } $\mathcal{C} = \{\mathbf{c}_j\}_{j=0}^{\mathfrak{R}-1}$, where $\mathfrak{R}=2^R$}
\REPEAT
    \STATE{Fix the codevectors, then update the encoding indexes using \eqref{eq:final enc}.}
    \STATE{Fix the encoding indexes, then update the codevectors using \eqref{eq:opt dec final}.}
\UNTIL{convergence}
\STATE{\textbf{output: } $\{\mathcal{R}_i\}_{i=0}^{\mathfrak{R}-1}$ , $\mathcal{C}=\{\mathbf{c}_j\}_{j=0}^{\mathfrak{R}-1}$ }
\end{algorithmic}
\end{algorithm}

\subsection{Practical Comparison} \label{subsec:pract comp}
\begin{figure*}[ht]
 \centering
 \subfigure[Encoding regions using the proposed COVQ-CS \eqref{eq:regions} in $\mathbb{R}^2$]{
  \includegraphics[width=2.3in]{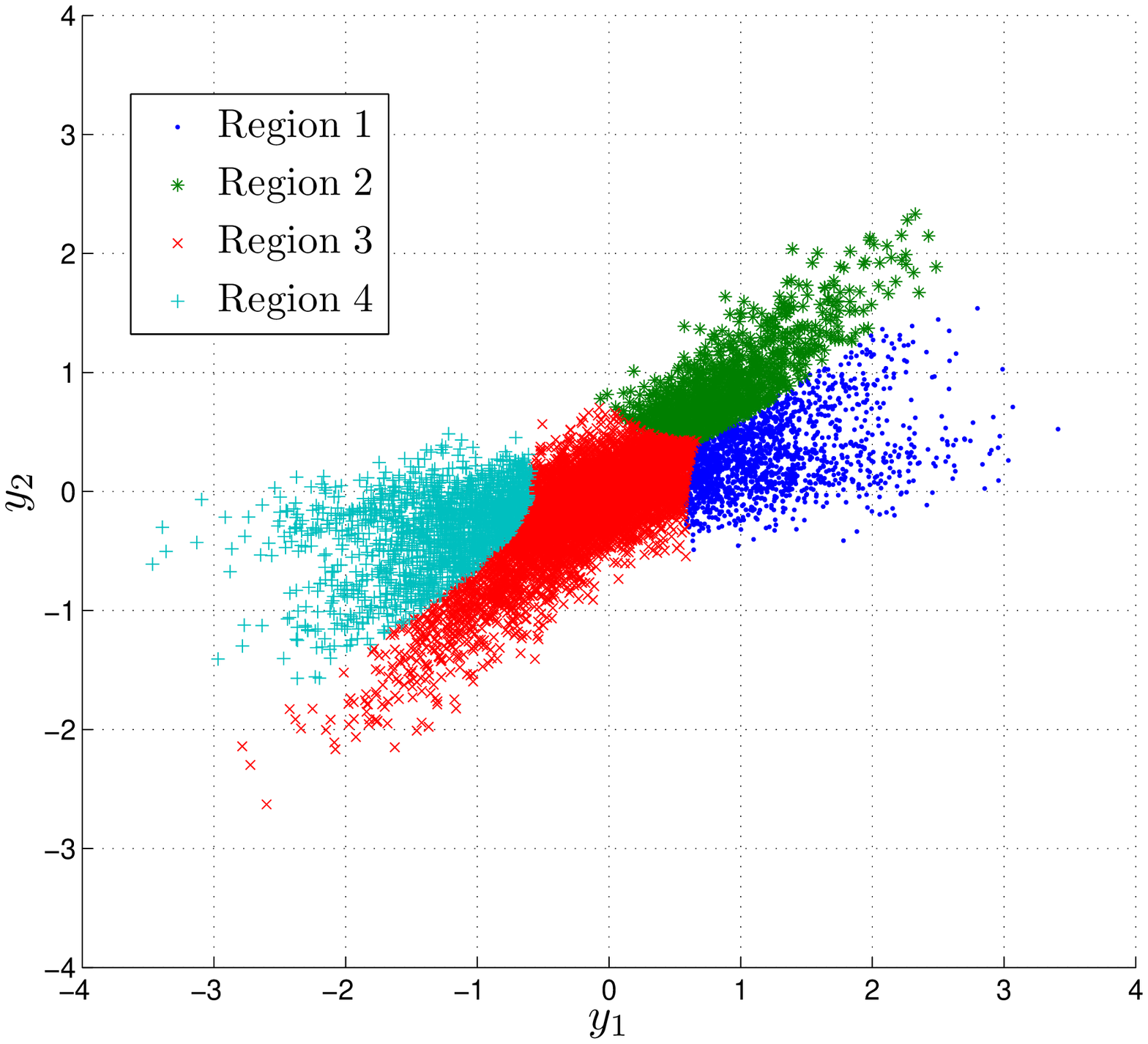}
  \label{fig:region_opt_2}
  }
 \subfigure[Codevectors designed by COVQ-CS in $\mathbb{R}^3$]{
  \includegraphics[width=2.3in]{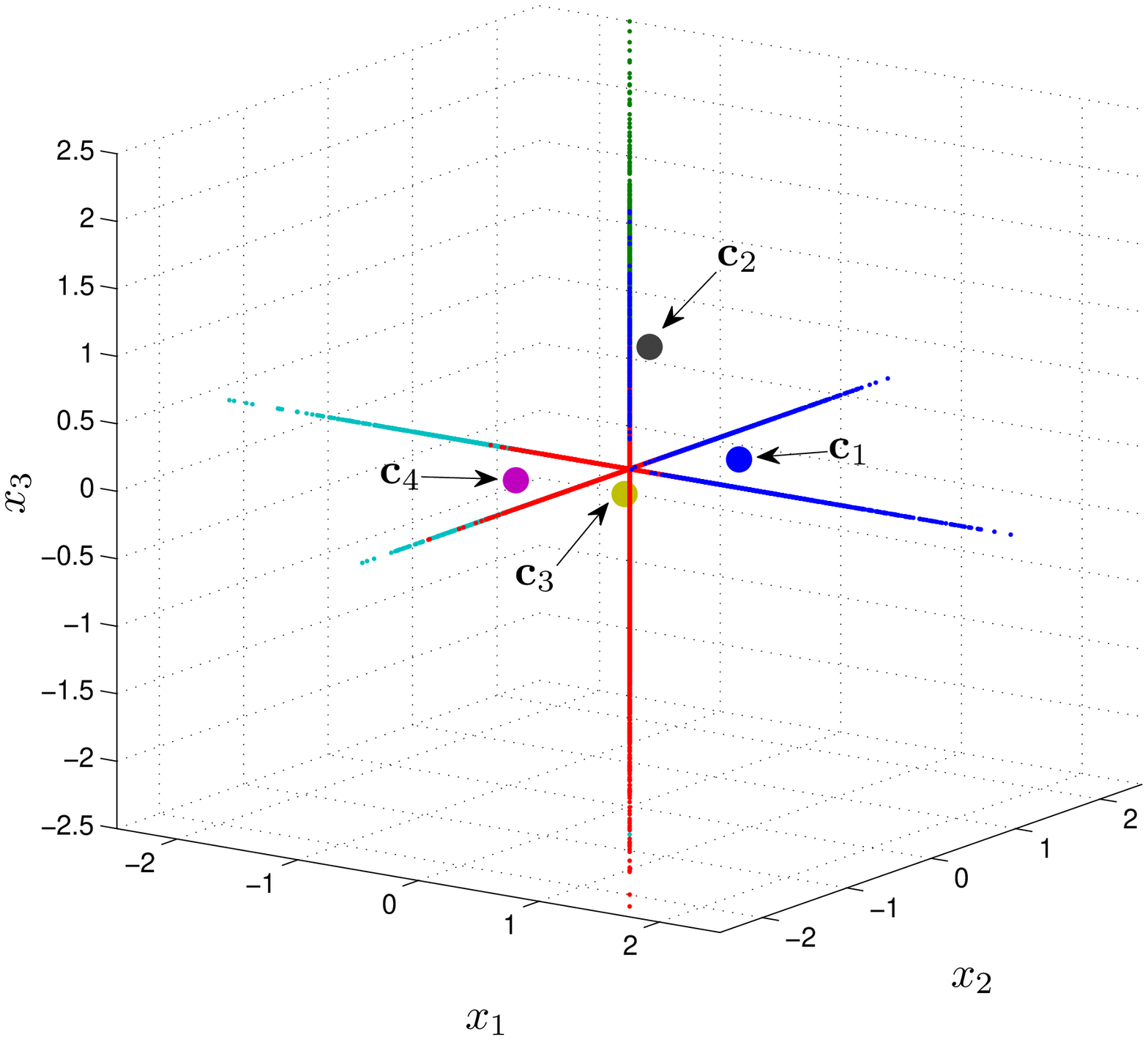}
  \label{fig:region_opt_3} \vspace{-0.4cm}
  }
  \subfigure[Encoding regions using the NNC-CS \eqref{eq:CO region ex} in $\mathbb{R}^2$]{
  \includegraphics[width=2.3in]{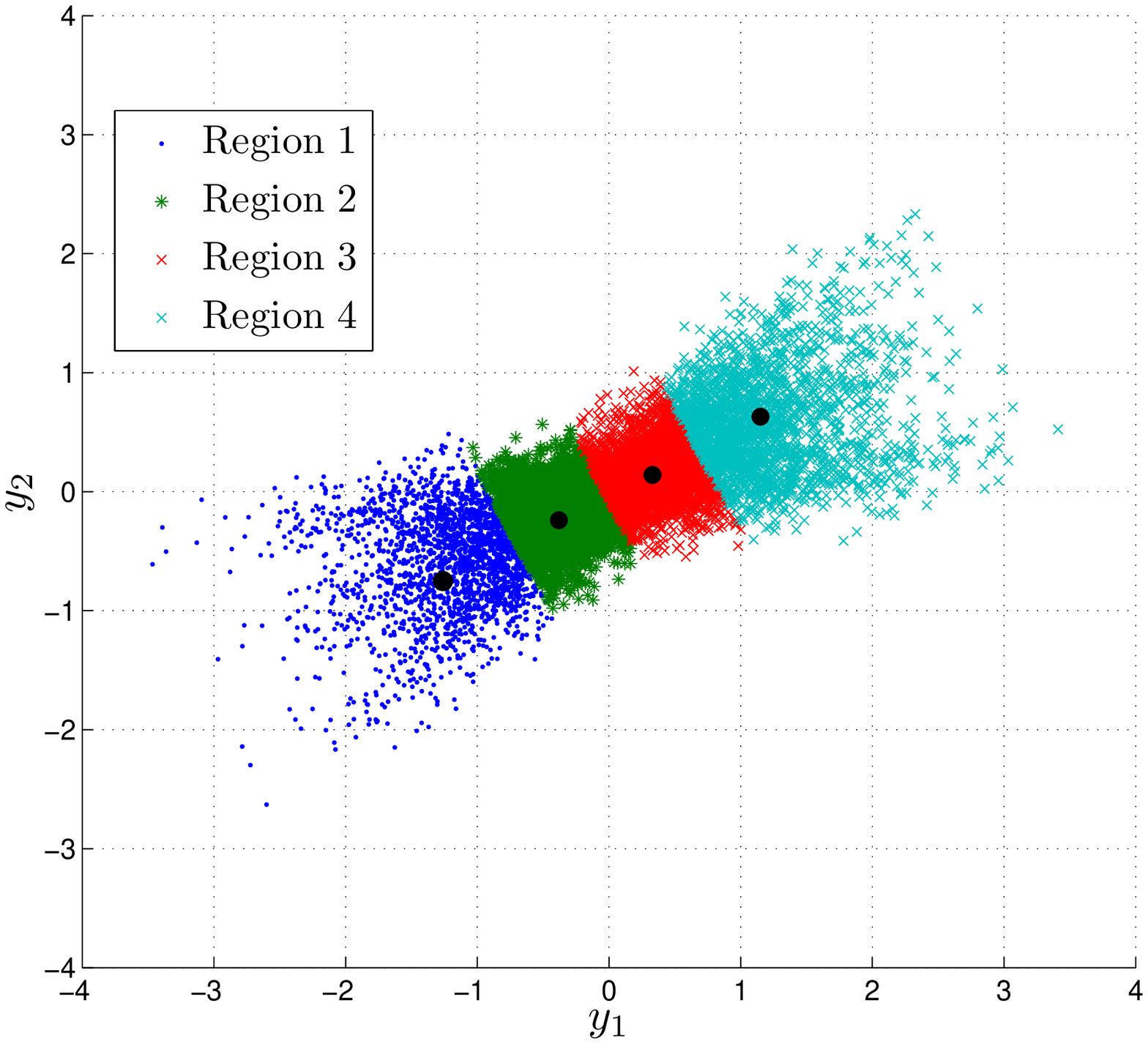}
  \label{fig:region_subopt_2}
  }
  \subfigure[Reconstructed Codevectors designed by NNC-CS in $\mathbb{R}^3$]{
  \includegraphics[width=2.3in]{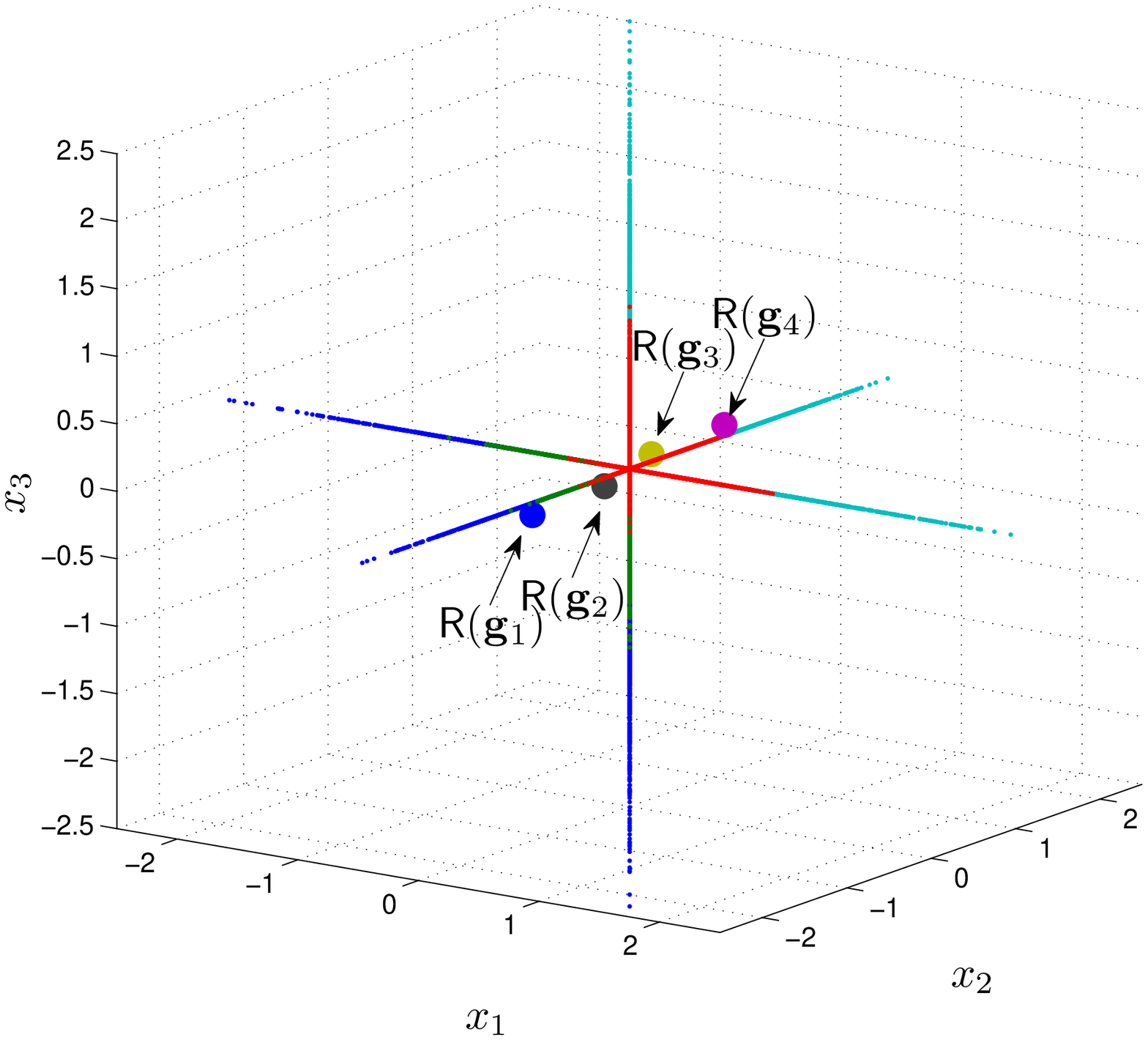}
  \label{fig:region_subopt_3} 
  }
 \caption{The qualitative behavior comparison of quantizer schemes: COVQ-CS and NNC-CS designed for a BSC with $\epsilon = 0.02$.}
  \label{fig:regions}
\end{figure*}

Here, we offer further insights into quantization aspects through the design of conventional \textit{nearest-neighbor coding} (NNC) as a representative of \figref{fig:subopt}, and the design of proposed COVQ-CS method as a representative of \figref{fig:decomposed VQ}. The NNC for CS is often considered as a benchmark for performance evaluations. 

The nearest-neighbor coding (NNC) for CS measurements is accomplished by designing a channel-optimized VQ for the input vector $\mathbf{Y}$ aiming to minimize the \textit{quantization distortion}, i.e., $\mathbb{E}[\|\mathbf{Y} \!-\! \widehat{\mathbf{Y}}\|_2^2]$, where $\widehat{\mathbf{Y}} \! \in \! \mathbb{R}^M$ is the quantizer decoder output as shown in \figref{fig:subopt}.\footnote[1]{See e.g. \cite{90:Farvardin} for more details regarding the design of channel-optimized VQ in a non-CS system model.} Considering the notations given for the \figref{fig:subopt}, the design procedure of the quantizer encoder and the quantizer decoder is as follows: for a quantization rate $R$ bits/vector, a fixed codebook $\mathcal{G} \!=\! \{\mathbf{g}_v \! \in \! \mathbb{R}^M\}_{v=0}^{\mathfrak{R}\!-\!1}$, with $\mathfrak{R}\!=\!2^R$, and channel transition probability $P(v|u)$, the optimized encoding region $\mathcal{R}_{u}^\star$ becomes
    \begin{equation} \label{eq:CO region ex}
    \begin{aligned}
        \mathcal{R}_{u}^\star  =  &\left\{\mathbf{y} \in \mathbb{R}^M  : \sum_{v=0}^{\mathfrak{R}-1} \left[P(v|u) - P(v|u')\right] \left\| \mathbf{g}_v \right\|_2^2 \leq \right.& \\
        & \left. 2 \mathbf{y}^\top \sum_{v=0}^{\mathfrak{R}-1} \left[P(v|u) - P(v|u') \right] \mathbf{g}_v , u \neq u' \in \mathcal{U} \right\},&
    \end{aligned}
    \end{equation}
where $\mathcal{U} \triangleq \{0,\ldots,2^R-1\}$ is the encoding index set.
    Now, for the given region \eqref{eq:CO region ex} and channel transition probability $P (j|i)$, the quantization MSE-minimizing codevectors satisfy
    \begin{equation} \label{eq:CO cent ex}
            \mathbf{g}_{v}^\star = \mathbb{E}[\mathbf{Y} | V=v]  , \hspace{0.5cm}  v \in \mathcal{U}.
    \end{equation}
    In order to design an encode-decoder pair using the NNC, an iterative algorithm can be used to alternate between \eqref{eq:CO region ex} and \eqref{eq:CO cent ex}. Finally, a CS reconstruction algorithm $\textsf{R}$ produces the reconstruction vector $\widehat{\mathbf{X}}$ from the quantizer decoder output $\widehat{\mathbf{Y}}$. We refer to this design method as NNC-CS.

\begin{ex} \label{ex:NNC}
    In this example, we illustrate how the COVQ-CS and NNC-CS design methods are different in shaping encoding regions (given that CS measurements are observed) and positioning codevectors (given that channel output index observed). For illustration purpose, we choose the input sparse vector dimension, measurement vector dimension and sparsity level as $N=3$, $M=2$ and $K=1$, respectively. The location of non-zero coefficient is drawn uniformly at random from $\{1,\ldots,N\}$, and its value is a standard Gaussian RV. For implementing the COVQ-CS via \algref{alg:Lloyd}, the MMSE estimator $\widetilde{\mathbf{x}}(\mathbf{y})$ (used in \eqref{eq:final enc}) is calculated via the closed-form solution given in \cite[eq. (27)]{09:Elad}. We generate $10^4$ realizations for $\mathbf{X}$ (and subsequently $\mathbf{Y}$), where measurement noise vector is drawn from $\mathcal{N}(\mathbf{0},\sigma_w^2 \mathbf{I}_M)$ with $\sigma_w^2 = 0.04$. Then, we fix the quantization rate at $R=2$ bits/vector and assume a BSC with cross-over probability $\epsilon = 0.02$. For implementing the NNC-CS, an iterative algorithm is used by alternating between encoding regions \eqref{eq:CO region ex} and codevectors \eqref{eq:CO cent ex}. Finally, a CS reconstruction algorithm $\textsf{R}: \mathbb{R}^M \rightarrow \mathbb{R}^N$ (here, we choose the same MMSE estimator used at the encoder of COVQ-CS) takes the NNC-CS codevectors and produces an estimate of the sparse source. In both NNC-CS and COVQ-CS schemes, the sensing matrix $\mathbf{\Phi}$ is chosen as
    \[ \mathbf{\Phi} = \left( \begin{array}{ccc}
        0.9924 & 0.8961 & 0.7201 \\
        0.1230 & 0.4439 & 0.6939 \\
    \end{array} \right).\]

    In \figref{fig:regions}, we qualitatively illustrate encoding regions and codevectors using the two designs. \figref{fig:region_opt_2} shows the samples of CS measurements classified by encoding regions of COVQ-CS in $\mathbb{R}^2$, i.e., \eqref{eq:final enc}, and \figref{fig:region_opt_3} shows the samples of $\mathbf{X}$ classified by the index of encoding regions (in the same color) together with the codevectors of COVQ-CS in $\mathbb{R}^3$, i.e., $\{\mathbf{c}_j\}_{j=1}^4$ in \eqref{eq:opt dec final}. \figref{fig:region_subopt_2} illustrates the encoding regions of NNC-CS, i.e., \eqref{eq:CO region ex}, together with codevectors $\{\mathbf{g}_v\}_{v=1}^4$ shown by black circles, and \figref{fig:region_subopt_3} shows the samples of the sparse source along with the codevectors of NNC-CS mapped to the 3-dimensional space using the CS reconstruction algorithm, i.e., $\textsf{R}(\{\mathbf{g}_v\}_{v=1}^4)$. From the samples in the measurement space, we observe that the entries of the CS measurements are highly correlated, in this particular example, due to a large mutual coherence of the sensing matrix ($\mu = 0.9533$). Hence, as shown in \figref{fig:region_subopt_2}, the codevectors designed by the NNC-CS (almost) lie on a single line. Although, in this case, the location of codevectors are optimized to minimize the quantization distortion, $\mathbb{E}[\|\mathbf{Y} - \widehat{\mathbf{Y}}\|_2^2]$, it is critical when the codevectors are mapped back to the source domain. From \figref{fig:region_subopt_3}, it is observed that the reconstructed codevectors, $\textsf{R}(\{\mathbf{g}_v\}_{v=1}^4)$, are not only situated (approximately) on one axis but also far (in Euclidean distance) to their corresponding source samples (shown in same color) resulting in a high end-to-end distortion. Further, if, for example, the codevector $\mathbf{g}_1$ is received as $\mathbf{g}_4$ due to channel noise, it produces a large end-to-end distortion. Using other experiments, in the case of noiseless channel, we observed the same trend in the location of reconstructed codevectors (using NNC-CS) on the source domain which also produces large MSE in terms of the average distance between source samples and their corresponding reconstructed codevectors. While this is the case in NNC-CS, it can be seen from \figref{fig:region_opt_2} that the encoding regions using COVQ-CS may not form convex sets (for example, region 3) unlike the ones using the NNC-CS. This is due to the fact that the region fixed by the rule \eqref{eq:regions} may not be a convex set in $\mathbf{y}$ due to non-linearity in $\widetilde{\mathbf{x}}(\mathbf{y})$. As a result, the COVQ-CS uses the measurement space more efficiently in order to reduce end-to-end distortion, $\mathbb{E}[\|\mathbf{X} - \widehat{\mathbf{X}}\|_2^2]$. It can be observed from \figref{fig:region_opt_3} that the COVQ-CS codevectors are located on different coordinates in the 3-dimensional source space to minimize the end-to-end source distortion. In addition, the codevectors are located such that the COVQ-CS design becomes more robust against channel noise which produces smaller end-to-end distortion unlike the NNC-CS design. For example, as shown in \figref{fig:region_opt_3}, if the codevector $\mathbf{c}_1$ is chosen as $\mathbf{c}_4$ at decoder due to channel noise, it provides much less end-to-end distortion than that of the NNC-CS. Numerical performance comparison between these two schemes will be made later in \secref{subsec:results} through different simulation studies.

\end{ex}
\subsection{Complexity of COVQ-CS} \label{subsec:VQ complexity}
We analyze the encoding computational complexity (time usage) as well as encoder-decoder memory complexity (space usage) for the COVQ-CS. For encoding computational complexity, we calculate the number of operations (in terms of FLOP\footnote[1]{Each addition, multiplication and comparison is represented by one floating point operation (FLOP).}) required for transmitting an encoded index over the channel based on \eqref{eq:final enc}. In addition, for memory complexity, we calculate the memory (in terms of float\footnote[2]{Float is considered as a single precision point unit.}) required for storing vector parameters at encoder and decoder.

The encoding complexity for computing the argument in \eqref{eq:final enc} requires one FLOP for calculating the subtraction as well as $2N\!-\!1$ FLOP's ($N$ multiplications and $N\!-\!1$ additions) for calculating the inner product in the second term. Thus, the total complexity for the full-search minimization at encoder is $2N 2^R$ FLOP's. Note that we do not consider the complexity of CS reconstruction algorithm since its calculation is required for all relevant quantizers for CS. Next, considering the argument in \eqref{eq:final enc}, the encoder needs one float to store the first constant term in \eqref{eq:final enc}, i.e., $\|\mathbf{c}_j\|_2^2$, and also $N$ floats to store the second term in \eqref{eq:final enc}, i.e., the codevector $\mathbf{c}_j$. Thus, the total encoding memory for full-search minimization is $(N\!+\!1)2^R$. It also follows that the decoder requires $N2^R$ floats to store $\mathbf{c}_j$ in \eqref{eq:opt dec}.

Using high-dimensional VQ and CS, the implementation of the quantizer encoder and decoder may not be feasible, both from computational complexity and from memory complexity viewpoints. The complexity can be reduced by exploiting sub-optimal approaches (with respect to \eqref{eq:e2e dist}) such as multi-stage VQ (MSVQ) which splits a single VQ into multiple VQ's at different stages. In the next section, we focus on the design of JSCC strategies for CS measurements using MSVQ.

\section{Joint Source-Channel MSVQ for CS} \label{sec:COMSVQ}
Taking advantage of VQ properties by addressing its encoding complexity effectively has led to development of multi-stage VQ (MSVQ). 

\subsection{System Description and Performance Criterion} \label{subsec:system model}
In this section, we give an account for the basic assumptions and models made about the investigated system depicted in \figref{fig:MSVQ}. We illustrate an $L$-stage VQ, where $L \geq 1$ is the maximum number of stages. Our MSVQ system model basically follows that of \cite{93:Phamdo}.
\begin{figure}
  \begin{center}
  \psfrag{x}[][][0.6]{$\mathbf{X}$}
  \psfrag{A}[][][0.85]{$\mathbf{\Phi}$}
  \psfrag{y}[][][0.6]{$\mathbf{Y}$}
  \psfrag{w}[][][0.6]{$\mathbf{W}$}
  \psfrag{Q_1}[][][0.65]{$\textsf{E}_1$}
  \psfrag{Q_2}[][][0.65]{$\textsf{E}_2$}
  \psfrag{Q_R}[][][0.65]{$\textsf{E}_L$}
  \psfrag{D_1}[][][0.65]{$\textsf{D}_1$}
  \psfrag{D_2}[][][0.65]{$\textsf{D}_2$}
  \psfrag{D_R}[][][0.65]{$\textsf{D}_L$}
  \psfrag{Pji}[][][0.55]{$P(j_1|i_1)$}
  \psfrag{Pji2}[][][0.55]{$P(j_2|i_2)$}
  \psfrag{PjiR}[][][0.55]{$P(j_L|i_L)$}
  \psfrag{C_1}[][][0.65]{$\mathcal{C}_1$}
  \psfrag{C_2}[][][0.65]{$\mathcal{C}_2$}
  \psfrag{C_R}[][][0.65]{$\mathcal{C}_L$}
  \psfrag{i_1}[][][0.6]{$I_1$}
  \psfrag{i_2}[][][0.6]{$I_2$}
  \psfrag{i_R}[][][0.6]{$I_L$}
  \psfrag{i_R-1}[][][0.6]{$I_{L-1}$}
  \psfrag{j_1}[][][0.6]{$J_1$}
  \psfrag{j_2}[][][0.6]{$J_2$}
  \psfrag{j_R}[][][0.6]{$J_L$}
  \psfrag{DMC}[][][0.65]{Channel}
  \psfrag{Dec}[][][0.65]{Decoder}
  \psfrag{x_h_1}[][][0.6]{$\widehat{\mathbf{X}}_1$}
  \psfrag{x_h_2}[][][0.6]{$\widehat{\mathbf{X}}_2$}
  \psfrag{x_h_R}[][][0.6]{$\widehat{\mathbf{X}}_L$}
  \psfrag{x_h}[][][0.6]{$\widehat{\mathbf{X}}$}
  \includegraphics[width=\columnwidth,height=4cm]{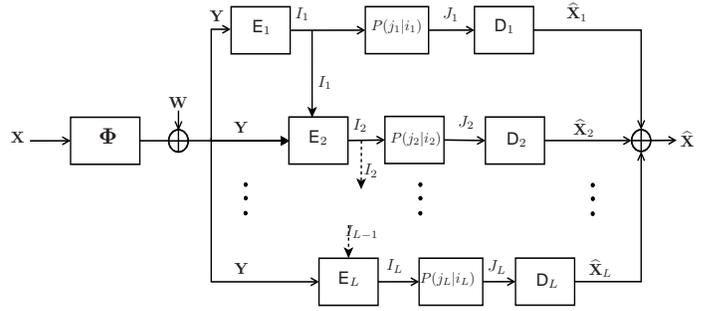}\\
  \caption{JSCC system model for CS measurements using MSVQ.}\label{fig:MSVQ}
  \end{center}
\end{figure}
More specifically, we consider the $l^{th}$ ($1 \leq l \leq L$) stage with allocated $R_l$ bits/vector, where $\sum_{l=1}^L R_l = R$, and $R$ is the total available quantization rate. Indeed, $R_l$ adjusts a trade-off between complexity and performance of MSVQ. A quantizer encoder, at stage $l$, accepts the measurement vector $\mathbf{Y}$ and the encoded index from the $(l-1)^{th}$ stage as inputs, then maps them into an integer index $i_l \in \mathcal{I}_l \triangleq \{0,\ldots,2^{R_l}-1\}$ with $|\mathcal{I}_l| \triangleq \mathfrak{R}_l = 2^{R_l}$. Therefore, the $l^{th}$--stage encoder is described by a mapping $\textsf{E}_l: \mathbb{R}^M \times \mathcal{I}_{l-1} \rightarrow \mathcal{I}_l$ such that
\begin{equation} \label{eq:MSVQ enc2}
   \textsf{E}_l\left(\mathbf{Y},I_{l-1}\right) = i_l, \hspace{0.4cm} \textrm{if } (\mathbf{Y} \in \mathcal{R}_{i_1}^{i_l}, \hspace{0.1cm} I_{l-1}=i_{l-1}),
\end{equation}
where $\mathcal{R}_{i_1}^{i_l} \triangleq \mathcal{R}_{i_{1}} \cap \ldots \cap \mathcal{R}_{i_l}$ is called the $l^{th}$--stage encoding region. The region $\mathcal{R}_{i_1}^{i_l}$ might be a connected set or union of some connected sets in $\mathbb{R}^M$. We also make the assumption that $\mathcal{I}_0 = \varnothing$.

The encoded index $I_l$ is transmitted over a DMC (independent of other channels) with transition probabilities
\begin{equation} \label{eq:DMC2}
    P(j_l|i_l) = \textrm{Pr} (J_l = j_l | I_l = i_l) , \hspace{0.25cm} i_l,j_l \in \mathcal{I}_l,
\end{equation}
where $J_l$ denotes the channel output at the $l^{th}$ stage.

Next, a decoder $\textsf{D}_l$ accepts the noisy index $J_l$, and provides an estimation of the quantization error according to an available codebook set. Formally, the $l^{th}$--stage decoder is defined by a mapping $\textsf{D}_l: \mathcal{I}_l \rightarrow \mathcal{C}_l$ where $\mathcal{C}_l$ denotes a codebook set consists of reproduction codevectors, i.e., $\mathcal{C}_l \triangleq \{\mathbf{c}_{j_l} \in \mathbb{R}^N\}_{j_l=0}^{\mathfrak{R}_l-1}$, thus
\begin{equation} \label{eq:MSVQ dec1}
   \textsf{D}_l(j_l) = \mathbf{c}_{j_l}, \hspace{0.4cm} \textrm{if } J_l=j_l , \hspace{0.2cm} j_l \in \mathcal{I}_l.
\end{equation}
We denote the output of the $l^{th}$ stage decoder by $\widehat{\mathbf{X}}_l = \mathbf{c}_{J_l}$, and the final reconstructed vector by $\widehat{\mathbf{X}} = \sum_{l=1}^L \widehat{\mathbf{X}}_l$.

We are interested in designing the quantizers in the system of \figref{fig:MSVQ} using the \textit{end-to-end MSE} criterion defined in \eqref{eq:e2e dist}. Nevertheless, it is not easy to find optimal encoders (by fixing the decoders) and decoders (by fixing encoders) for all the stages jointly with respect to minimizing \eqref{eq:e2e dist}. Therefore, we define a new performance criterion as
\begin{equation} \label{eq:MSE MSVQ}
    D_l \triangleq \mathbb{E}[\|\mathbf{X} - \sum_{t=1}^l \widehat{\mathbf{X}}_t \|_2^2] , \hspace{0.15cm} l=1,\ldots,L .
\end{equation}

Using the performance criterion $D_l$ in \eqref{eq:MSE MSVQ}, we assume that the $l^{th}$ stage only observes the previous $(l-1)$ stages. Applying $D_l$, we derive necessary encoding and decoding policies for optimality (with respect to \eqref{eq:MSE MSVQ}) at stage $l$ ($1 \leq l \leq L$). Then, encoder-decoder pairs at the next stages are \textit{sequentially} designed one after another. Using the sequential optimization at stage $l$, we assume that the subsequent codevectors are populated with zero. This assumption means that the sequential design is sub-optimal with respect to \eqref{eq:e2e dist}, and the resulting conditions would lead to neither global nor local minimum of the end-to-end MSE. However, it provides better performance compared to the schemes which only consider quantization distortion at each stage separately.

\subsection{Optimality Conditions for MSVQ Encoder and Decoder} \label{sec:MSVQ}
In this section, we develop encoding and decoding principles for the $l^{th}$ ($1 \leq l \leq L$) stage of the MSVQ system shown \figref{fig:MSVQ}. Following the arguments of \secref{sec:design COVQ}, we first assume that decoder codevectors $\{\mathbf{c}_{j_l}\}_{j_l=0}^{\mathfrak{R}_l-1}$ and all encoding regions/codevectors at previous $l-1$ stages are fixed and known, then we find necessary optimal encoding regions with respect to minimizing $D_l$ in \eqref{eq:MSE MSVQ} in \secref{subsec:encoder MSVQ}. Second, we fix the encoding regions $\{\mathcal{R}_{i_1}^{i_l}\}$, and then derive necessary optimal codevectors in \secref{subsec:decoder MSVQ}. Finally, in \secref{subsec:training MSVQ}, we combine these necessary optimal conditions to develop a practical MSVQ design algorithm referred to as channel-optimized MSVQ for CS (COMSVQ-CS).

\subsubsection{Optimal Encoder} \label{subsec:encoder MSVQ}
In order to derive encoding regions $\{\mathcal{R}_{i_1}^{i_l}\}$, we fix the codevectors $\{\mathbf{c}_{j_l}\}_{j_l=0}^{\mathfrak{R}_l-1}$ and all the codevectors at previous stages. First, let us define
\begin{equation} \label{eq:helping param}
    D_l(\mathbf{y}, \mathbf{i}_1^{l}) \triangleq \mathbb{E} [\|\mathbf{X} - \sum_{t=1}^l \widehat{\mathbf{X}}_t \|_2^2 \big| \mathbf{Y=y} , \mathbf{I}_1^l = \mathbf{i}_1^l ], \hspace{0.15cm} 1 \leq l \leq L.
\end{equation}

Now, $D_l$ in \eqref{eq:MSE MSVQ} can be rewritten as
\begin{equation} \label{eq:rewrite MSE_MSVQ}
\begin{aligned}
    D_l &\triangleq \mathbb{E}[\|\mathbf{X} -  \sum_{t=1}^l \mathbf{c}_{J_t} \|_2^2]&\\
    &\stackrel{(a)}{=} \sum_{i_1,\!\ldots\!,i_l} \int_{\mathcal{R}_{i_1}^{i_l}}  D_l(\mathbf{y}, \mathbf{i}_1^{l})  f(\mathbf{y}) d\mathbf{y} ,&
\end{aligned}
\end{equation}
where $(a)$ follows from marginalization of $D_l$ over $\mathbf{Y}$ and $I$, and the fact that $\textrm{Pr} \{\mathbf{I}_1^l=\mathbf{i}_1^l | \mathbf{Y=y}\} = 1$, $\forall \mathbf{y} \in \mathcal{R}_{i_1}^{i_l}$ and otherwise the probability is zero.

Thus, $\forall i_1 ,\ldots, i_{l-1}$, the optimized index, denoted by $i_l^\star$, is attained by \eqref{eq:proof prim enc}, where $(a)$ follows from the Markov chain $\mathbf{X} \rightarrow (\mathbf{Y},\mathbf{I}_1^{t-1}) \rightarrow I_t$ ($\forall t \in \{1,\ldots,l\}$), hence, $\mathbb{E} \left[\|\mathbf{X}\|_2^2 | \mathbf{Y\!=\!y}, \mathbf{I}_1^l\!=\!\mathbf{i}_1^l \right] = \mathbb{E} \left[\|\mathbf{X}\|_2^2 | \mathbf{Y\!=\!y}\right]$ which is pulled out of the optimization. Also, $(b)$ follows from the Markov chain $(\mathbf{Y}, \mathbf{I}_1^l) \rightarrow  I_t \rightarrow  \mathbf{c}_{J_t}$, $\forall t \in \{1,\ldots,l\}$.

\begin{figure*}[!t]
\normalsize
\setcounter{MYtempeqncnt}{\value{equation}}
\begin{equation} \label{eq:proof prim enc}
\begin{aligned}
    i_l^\star &= \textrm{arg }\underset{i_l \in \mathcal{I}_l}{\textrm{min }} D_l(\mathbf{y}, \mathbf{i}_1^{l})& \\
    &= \textrm{arg }\underset{i_l \in \mathcal{I}_l}{\textrm{min }} \left\{ \mathbb{E} \left[\|\mathbf{X}\|_2^2 | \mathbf{Y=y}, \mathbf{I}_1^l=\mathbf{i}_1^l \right]  + \mathbb{E}\left[\|\mathbf{c}_{J_1}\!+\!\ldots\!+\! \mathbf{c}_{J_l}\|_2^2 | \mathbf{Y\!=\!y} , \mathbf{I}_1^l \!=\! \mathbf{i}_1^l \right] - 2\mathbb{E}\left[\mathbf{X}^\top (\mathbf{c}_{J_1}\!+\!\ldots\!+\! \mathbf{c}_{J_l}) | \mathbf{Y\!=\!y}, \mathbf{I}_1^l\!=\!\mathbf{i}_1^l\right] \right\} & \\
    &\stackrel{(a)}{=} \textrm{arg }\underset{i_l \in \mathcal{I}_l}{\textrm{min }} \left\{ \mathbb{E}[\|\mathbf{c}_{J_1}\!+\!\ldots\!+\! \mathbf{c}_{J_l}\|_2^2  | \mathbf{Y=y}, \mathbf{I}_1^l=\mathbf{i}_1^l] - 2\mathbb{E}[\mathbf{X}^\top (\mathbf{c}_{J_1}\!+\!\ldots\!+\! \mathbf{c}_{J_l})  | \mathbf{Y=y},\mathbf{I}_1^l=\mathbf{i}_1^l ] \right\}& \\
    &\stackrel{(b)}{=} \textrm{arg }\underset{i_l \in \mathcal{I}_l}{\textrm{min }} \left\{ \mathbb{E}[\|\mathbf{c}_{J_l}\|_2^2 \big | I_l=i_l] + 2 \sum_{t=1}^{l-1} \mathbb{E}[\mathbf{c}_{J_l}^\top | I_l=i_l] \mathbb{E}[\mathbf{c}_{J_t} | I_t=i_t] - 2\mathbb{E}[\mathbf{X}^\top \big | \mathbf{Y=y}] \mathbb{E}[\mathbf{c}_{J_l} \big | I_l=i_l] \right\}&
\end{aligned}
\end{equation}
\setcounter{equation}{\value{MYtempeqncnt}}
\hrulefill
\end{figure*}
\setcounter{equation}{23}

Introducing transition probabilities \eqref{eq:DMC2} and the MMSE estimator \eqref{eq:MMSE sparse}, the last equality in \eqref{eq:proof prim enc} is expressed as
\begin{equation} \label{eq:final enc MSVQ}
\begin{aligned}
    i_l^\star &\!=\! \textrm{arg }\underset{i_l \in \mathcal{I}_l}{\textrm{min}} \left\{\sum_{j_l=0}^{\mathfrak{R}_l-1} P(j_l|i_l) \left\| \mathbf{c}_{j_l} \right\|_2^2 \!-\! 2 \widetilde{\mathbf{x}}(\mathbf{y})^\top \sum_{j_l=0}^{\mathfrak{R}_l-1} P(j_l|i_l) \mathbf{c}_{j_l} \right.& \\
    &\left.+ 2 \sum_{j_l=0}^{\mathfrak{R}_l-1} \sum_{t=1}^{l-1}  \sum_{j_t=0}^{\mathfrak{R}_t-1} P(j_l|i_l) P(j_t|i_t) \mathbf{c}_{j_l}^\top \mathbf{c}_{j_t} \right\} .&
\end{aligned}
\end{equation}

\begin{rem}
    Comparing the optimized encoding index for MSVQ for CS in \eqref{eq:final enc MSVQ}, with that of the VQ for CS in \eqref{eq:final enc}, it can be seen that the third term in \eqref{eq:final enc MSVQ} is due to imposing multi-stage structure on the original VQ. As $L=1$, this term vanishes and the resulting expression coincides with \eqref{eq:final enc}.
\end{rem}

\subsubsection{Optimal Decoder} \label{subsec:decoder MSVQ}
In order to derive codevectors $\{\mathbf{c}_{j_l}\}_{j_l=0}^{\mathfrak{R}_l-1}$, we fix encoding regions $\{\mathcal{R}_{i_1}^{i_l}\}$ and all prior codebook sets. Therefore, applying $D_l$ in \eqref{eq:MSE MSVQ}, it is straightforward to show that the optimal $l$--stage codevectors, denoted by $\{\mathbf{c}_{j_l}^\star\}_{j_l=1}^{\mathfrak{R}_l-1}$, are obtained as
\begin{equation} \label{eq:codevector MSVQ}
\begin{aligned}
    \mathbf{c}_{j_l}^\star     &= \mathbb{E}[\mathbf{X} - \sum _{t=1}^{l-1} \mathbf{c}_{J_t} | J_l = j_l] , \hspace{0.4cm} j_l \in \mathcal{I}_l.&
\end{aligned}
\end{equation}
Similar to the steps taken in \eqref{eq:opt dec final}, the codevectors \eqref{eq:codevector MSVQ} can be parameterized in terms of encoding regions, channel transition probabilities and MMSE estimation. Here, for the sake of analysis, we only provide closed-form codebook expressions for $L=2$ which are given by
\begin{equation*} \label{eq:MSVQ codevector param}
\begin{aligned}
    &\mathbf{c}_{j_1}^\star \!=\! \frac{\sum_{i_1} P(j_1|i_1) \int_{\mathcal{R}_{i_1}} \widetilde{\mathbf{x}}(\mathbf{y}) f(\mathbf{y}) d\mathbf{y} }{\sum_{i_1} P(j_1|i_1) \int_{\mathcal{R}_{i_1}} f(\mathbf{y}) d\mathbf{y}},& \\
    &\mathbf{c}_{j_2}^\star \!=\! \frac{\sum_{i_1,i_2} P(j_2|i_2) \int_{\mathcal{R}_{i_1}^{i_2}} \left(\widetilde{\mathbf{x}}(\mathbf{y}) \!-\! \sum_{j_1} P(j_1|i_1) \mathbf{c}_{j_1} \right) f(\mathbf{y}) d\mathbf{y}}{\sum_{i_1,i_2}  P(j_2|i_2) \int_{\mathcal{R}_{i_1}^{i_2}} f(\mathbf{y}) d\mathbf{y}}.&
\end{aligned}
\end{equation*}
Finally, we note that when $L\!=\!1$, the condition \eqref{eq:codevector MSVQ} simplifies into \eqref{eq:opt dec final}.
\subsubsection{Training Algorithm} \label{subsec:training MSVQ}

Similar to \algref{alg:Lloyd}, we can develop a practical method for training channel-optimized MSVQ for CS, coined COMSVQ-CS, summarized in \algref{alg:Lloyd MSVQ}. Similar remarks, as stated for \algref{alg:Lloyd}, can be also considered for implementing \algref{alg:Lloyd MSVQ} with the difference that convergence in step (8) may be checked by tracking the distortion $D_l$, and terminate the iterations when the relative improvement is small enough. Furthermore, in order to calculate the codevector $\mathbf{c}_{j_l}$ ($j_l \in \mathcal{I}_l$) in \eqref{eq:codevector MSVQ}, we use Monte-Carlo method by first generating a set of finite \textit{training vectors} $\mathbf{X}$, with known pdf, and then calculating the vector $\mathbf{X} - \sum _{t=1}^{l-1} \mathbf{c}_{J_t}$. Finally, we average over those vectors that have resulted the index $J_l=j_l$.

\begin{algorithm}
\caption{ COMSVQ-CS: Practical training algorithm for the $l^{th}$ stage ($1 \leq l \leq L$) of MSVQ}\label{alg:Lloyd MSVQ}
\begin{algorithmic}[1]
\STATE{\textbf{input:} measurement vector $\mathbf{y}$, channel probabilities: $P(j_l|i_l)$ from \eqref{eq:DMC2}, bit budget: $R_l$}
\STATE{\textbf{compute:} $\widetilde{\mathbf{x}}(\mathbf{y})$ in \eqref{eq:MMSE sparse}}
\STATE{\textbf{initialize: } $\mathcal{C}_l = \{\mathbf{c}_{j_l}\}_{j_l=0}^{\mathfrak{R}_l-1}$ with $\mathfrak{R}_l = 2^{R_l}$}
\REPEAT
    \STATE{$\forall i_1,\ldots, i_{l-1}$, $\forall \mathbf{c}_{j_1},\ldots, \mathbf{c}_{j_{l-1}}$}
    \STATE{Fix the codebooks of all prior stages, then update encoding indexes (regions) for the $l^{th}$ stage using \eqref{eq:final enc MSVQ}.}
    \STATE{Fix the encoding indexes (regions) of all prior stages, then update the codevectors for the $l^{th}$ stage using \eqref{eq:codevector MSVQ}.}
\UNTIL{convergence}
\STATE{\textbf{output: } $\mathcal{C}_l = \{\mathbf{c}_{j_l}\}_{j_l=0}^{\mathfrak{R}_l-1}$ and $\{\mathcal{R}_{i_1}^{i_l}\}$ }
\end{algorithmic}
\end{algorithm}

\subsection{Complexity of COMSVQ-CS} \label{cs:complexity COMSVQ-CS}
In order to calculate the MSVQ encoder complexity, we calculate the number of operations at the encoder based on \eqref{eq:final enc MSVQ}. Here, the computational complexity of CS reconstruction algorithm is not considered. We consider the argument of \eqref{eq:final enc MSVQ} which requires two FLOP's for the subtraction and addition, and also $2N-1$ FLOP's for computing the second inner product term. Note that the first constant term and the third inner product term can be computed offline, and they are not counted in our complexity analysis. Thus, in total, the COMSVQ-CS encoder requires $(2N + 1) \sum_{l=1}^L 2^{R_l}$ operations, where $R_l$ ($l=1,\ldots,L$) is the quantization rate available at a $l^{th}$ stage and $L$ is total number of stages such that $\sum_{l=1}^L R_l = R$. 

It can be also shown that at stage $l$, the encoder requires one float to store the first term in \eqref{eq:final enc MSVQ}, i.e., $\|\mathbf{c}_{j_l}\|_2^2$, $N$ floats to store the second term in \eqref{eq:final enc MSVQ}, i.e., $\mathbf{c}_{j_l}$, and also $l-1$ floats for storing the third term in \eqref{eq:final enc MSVQ}. Therefore, considering $L$ stages, the total encoding memory of the COMSVQ-CS is $\sum_{l=1}^L (N+\textit{l})2^{R_l}$. Now, we consider the decoder memory complexity. Each decoder at stage $l$ requires $N2^{R_l}$ floats to store the codevector $\mathbf{c}_{j_l}$ considering the fact that the memory for storing the codebooks of previous stages has been already calculated. Hence, the decoder storage memory is $N \sum_{l=1}^L 2^{R_l}$ floats.

By splitting the original VQ into stages, the computational complexity as well as memory complexity can be considerably reduced. Therefore, a practical, however sub-optimal, implementation of COVQ-CS is feasible at high quantization rate and dimension.

\section{Experiments and Results} \label{sec:numerical}
In this section, we evaluate the performance of the proposed designs COVQ-CS (\algref{alg:Lloyd}) and COMSVQ-CS (\algref{alg:Lloyd MSVQ}). Through simulations, we compare their performances with the lower-bounds developed in \secref{subsec:bound COVQ-CS} along with existing quantizers used for CS. We consider three quantizers following the system model of \figref{fig:subopt}. They are as follows.
\begin{itemize}
    \item \textit{Nearest-Neighbor Coding for CS (NNC-CS)}:\footnote[1]{Here, with abuse of notation we use the term \textit{nearest-neighbor coding} in the presence of channel noise instead of \textit{weighted nearest-neighbor coding}.} The NNC-CS design method has been discussed in \secref{subsec:pract comp}. Note that this scheme has the same complexity order as that of the COVQ-CS.

    \item \textit{Multi-Stage Nearest-Neighbor Coding for CS (MSNNC-CS)}: Using multi-stage structure for NNC-CS leads to the design of MSNNC-CS. The quantizer encoding-decoding conditions using this design are given in \cite{93:Phamdo} for a non-CS system model. The encoding complexity order of MSNNC-CS is the same as that of the COMSVQ-CS.
    \item \textit{Basis Pursuit DeQuantizing (BPDQ)} \cite{11:Jacques}: Using this method, the encoder uniformly scalar-quantizes CS measurements, and the BPDQ  algorithm \cite{11:Jacques} reconstructs the sparse source (from the quantized measurements) by solving the following convex optimization problem
        \begin{equation} \label{eq:bpdq}
            \mathbf{x}^\star = \textrm{arg }\underset{\mathbf{x} \in \mathbb{R}^N}{\textrm{min }} \|\mathbf{x}\|_1 \text{    s.t.  } \|\widehat{\mathbf{y}}  - \mathbf{\Phi x} \|_p \leq \gamma,
        \end{equation}
        where $\widehat{\mathbf{y}}$ is the quantized vector, $p > 2$ and $\gamma > 0$ is chosen to satisfy some fidelity constraint, e.g., quantization error power. Note that the encoder computational complexity is of order $\mathcal{O}(2^{R/M})$. In the design of uniform scalar quantizer for the BPDQ scheme, the choice of lower- and upper-boundaries for quantization is important, leading to different saturation errors \cite{11:Laska}. In order to choose the end-points for uniform quantization of CS measurements, we generate random samples of CS measurement vectors according to the distribution of the sparse source, sensing matrix, and the measurements noise. Then, the upper quantization boundary is selected as the maximum value among the amplitudes of the generated sample entries of the measurement vector. The lower quantization boundary is also selected as the opposite value of the upper-boundary. Using such simple approach, we mainly reduce the effect of the saturation error.
\end{itemize}
The following scheme follows the system model of \figref{fig:decomposed VQ}.
\begin{itemize}
    \item \textit{Support Set Coding (SSC)}: In the SSC method, the reconstructed support set of $\widetilde{\mathbf{x}}(\mathbf{y})$ is transmitted using $\log_2 {N \choose K}$ bits, and then the $K$ largest coefficients (in magnitude) within the reconstructed support set are scalar-quantized to their nearest neighbor codepoints using $R - \log_2 {N \choose K}$ bits. Here, we use codepoints optimized for a standard Gaussian RV using the LBG algorithm \cite{80:LBG}. Notice that when the non-zero coefficients of the sparse source vector are drawn according to an i.i.d. standard Gaussian distribution (which is the case in our simulations), the optimized LBG-based codepoints  minimize the distortion per non-zero component of the sparse source. It is straightforward that the encoding complexity of the SSC is of order $\mathcal{O}(2^{\left(R-\log_2 {N \choose K}\right)/K})$, or equivalently $\mathcal{O}(2^{R/K})$ at high quantization rate. \footnote[1]{In the spirit of reproducible results, we provide MATLAB codes for simulation of the AbS-based quantizers in the following website: www.ee.kth.se/$\sim$amishi/reproducible$\_$research.html.}
 \end{itemize}

\subsection{Experimental Setup} \label{subsec:setup}
We quantify the performance using normalized MSE (NMSE) defined as
\begin{equation} \label{eq:NMSE}
    \textrm{NMSE} \triangleq \frac{\mathbb{E}[\|\mathbf{X}-\widehat{\mathbf{X}}\|_2^2]}{\mathbb{E}[\|\mathbf{X}\|_2^2]}.
\end{equation}
In principle, the numerator of NMSE in \eqref{eq:NMSE} is computed by sample averaging over generated realizations of $\mathbf{X}$ using Monte-Carlo simulations, and the denominator can be calculated exactly under the assumptions of our simulation setup.

In addition, in order to measure the level of under-sampling, we define the measurement rate $0 < \alpha \leq 1$ as $\alpha \triangleq M/N$.

Our simulation setup includes the following steps:
\begin{enumerate}
  \item For given values of sparsity level $K$ (assumed known in advance) and input vector size $N$, choose $\alpha$, and round the number of measurements $M$ to its nearest integer.
  \item Randomly generate a set of exactly $K$-sparse vector $\mathbf{X}$, where the support set $\mathcal{S}$ with $|\mathcal{S}| = K$ is chosen uniformly at random over the set $\{1,2,\ldots,N\}$. The non-zero coefficients of $\mathbf{X}$ are i.i.d. and drawn from standard Gaussian source $\mathcal{N}(0,1)$; Hence $\mathbb{E}[\|\mathbf{X}\|_2^2] = K$.
 
\item We let the elements of the sensing matrix be $\mathbf{\Phi}_{ij}\! \overset{\textrm{iid}}{\sim} \!\mathcal{N} (0,1/M)$, and normalize its columns to unit-norm. Once $\mathbf{\Phi}$ is generated, it remains fixed and known globally.
  \item Compute linear measurements $\mathbf{Y= \Phi X + W}$ for each sparse data vector where $\mathbf{W} \sim \mathcal{N}(\mathbf{0},\sigma_w^2 \mathbf{I}_M)$.
    
     \item We Choose the total quantization rate $R$, and assume a BSC with bit cross-over probability $0 \leq \epsilon \leq 0.5$ specified by
      \begin{equation} \label{eq:BSC}
          P(k | l) = \epsilon^{H_\mathfrak{R}(k,l)} (1 - \epsilon)^{\mathfrak{R} - H_\mathfrak{R}(k,l)}, \hspace{0.1cm} \mathfrak{R} = 2^R,
      \end{equation}
      where $0 \leq \epsilon \leq 1/2$ represents bit cross-over probability (assumed known), and $H_\mathfrak{R}(k,l)$ denotes the Hamming distance between $\mathfrak{R}$-bit binary codewords representing the channel input and output indexes $k$ and $l$. The capacity of BSC (in bits per channel use) with bit cross-over probability $\epsilon$ is equivalent to
      \begin{equation} \label{eq:capacity BSC}
        C = 1 + \epsilon \log_2(\epsilon) + (1-\epsilon)\log_2(1-\epsilon).
      \end{equation}

  \item Apply the quantization algorithms on the generated data $\mathbf{Y}$, and assess NMSE by averaging over all data.
  \item \textit{Practical necessity:}  In our proposed COVQ-CS and COMSVQ-CS design algorithms, it is required to calculate the MMSE estimator $\widetilde{\mathbf{x}}(\mathbf{y})$, e.g. in \eqref{eq:final enc} and \eqref{eq:final enc MSVQ}. Implementing the Bayesian MMSE estimator, or in other words, calculating the conditional mean $\mathbb{E}[\mathbf{X}| \mathbf{Y=y}]$, has been studied in \cite{07:Larsson,08:Ji,09:Elad,10:Protter,12:Kun} which can be derived approximately or exactly under certain assumptions. Although the MMSE estimator can be implemented for low-dimensional vectors (as used in \exampref{ex:NNC}), as the dimension grows, its complexity increases exponentially. Thus, for the sake of complexity, we will approximate $\widetilde{\mathbf{x}}(\mathbf{y})$ using the output of a practically realizable CS reconstruction algorithm. Considering the case that a $\ell_1$-norm minimization-based convex reconstruction also suffers from high complexity $\mathcal{O}(N^3)$ for a high dimension $N$, we choose the simple orthogonal matching pursuit (OMP) greedy algorithm \cite{07:Tropp} as a CS reconstruction where its computational complexity is $\mathcal{O}(K^3 \!+\! K^2 M \!+\! KMN)$. We show that using the OMP algorithm, we can obtain reasonable performance. The OMP is used as the approximation of the MMSE estimator (at the encoder side) for COVQ-CS, COMSVQ-CS and SSC schemes as well as the realization of CS reconstruction algorithm (at the decoder side) for NNC-CS and MSNNC-CS methods.
\end{enumerate}

\subsection{Experimental Results} \label{subsec:results}

In our simulations, we generate $10^6$ realizations of the input sparse vector $\mathbf{X}$ (correspondingly $\mathbf{Y}$) for the training algorithms as well as performance assessments using Monte-Carlo simulations. We evaluate the performance of the competing schemes in terms of number of CS measurements ($\alpha$), total quantization rate ($R$) and channel condition $(\epsilon)$. It should be also mentioned that the training algorithms are performed at each value on x-axis, i.e., $\alpha$, $R$ and $\epsilon$.

In our first experiment, we assume that the measurement noise and channel cross-over probability are negligible, i.e., $\sigma_w^2 = 0$ and $\epsilon = 0$. \footnote[2]{With abuse of notation, we still use the term COVQ-CS and COMSVQ-CS when channel is noiseless ($\epsilon=0$).} In \figref{fig:MSE_FOM_10}, with the simulation setup ($N=12,K=2,R=12 \text{ bits/vector}$), we vary measurement rate $\alpha = M/N$, and compare the performance (NMSE) of the quantizers along with the lower-bound \eqref{eq:simple lb}. We use 2-stage VQ with equal quantization rates. For implementing the BPDQ decoder, we select $p=3$ in \eqref{eq:bpdq} (the choice of $p$ is experimentally verified to achieve the best performance), and $\gamma$ is chosen according to \cite[eq. (7)]{11:Jacques}, then a standard convex solver is used to find the optimal solution of \eqref{eq:bpdq}. Let us first investigate the behavior of the full search COVQ-CS and NNC-CS quantizer design schemes in \figref{fig:MSE_FOM_10}. At a fixed quantization rate $R$, increasing the number of measurements improves the CS reconstruction performance, hence the end-to-end MSE decreases. Since quantized transmission distortion $D_q$ is fixed, NMSE would saturate ultimately. As expected, the proposed COVQ-CS design method gives the best performance, and at high measurement rates, it approaches the lower-bound \eqref{eq:simple lb}. This is due to the fact that, at high measurement rate regime, the CS distortion $D_{cs}$ becomes negligible and the source vector can be precisely recovered from the measurements; therefore, COVQ-CS approaches the distortion rate function for the sparse source. Note that the performance gain using the COVQ-CS design scheme is obtained at the expense of computational and memory complexity of order $\mathcal{O}(2^{R})$. Using the sub-optimal COMSVQ-CS scheme, the complexity is decreased to the order $\mathcal{O}(2^{R/2})$ although its performance is slightly declined compared to COVQ-CS. Among multi-stage structured methods, the COMSVQ-CS performs better than MSNNC-CS since it takes end-to-end MSE through its design procedure. Also, it can closely follow the behavior of the COVQ-CS at low to moderate ranges of measurement rates since at this regime, the performance is mostly influenced by the CS reconstruction distortion. However, as the measurement rate increases, the gap between the performance of the COMSVQ-CS and COVQ-CS becomes larger. The gap can be made smaller if we use higher quantization rates at the first stage while keeping the total quantization rate fixed, however, this imposes more encoding complexity to the system. It can be also seen that the SSC scheme performs poorer than COVQ-CS and MSVQ-CS, while its encoding complexity grows at most like $O(2^{R/2})$. The behavior of the BPDQ, however, is different: increasing number of measurements, on one hand, facilitates a more precise reconstruction. On the other hand, it reduces quantization rate since each measurement entry is quantized using $R/M$ bits. Hence, the performance curve of BPDQ reaches a minimum point, and then takes an upward trend which also complies with the fact of CS and quantization compression regimes \cite{12:Laska}. Note that the BPDQ has the least computational complexity among the competing techniques varying from the order of $\mathcal{O}(2^{R/3})$ to $\mathcal{O}(2^{R/12})$.

\begin{figure}
  \begin{center}
  \includegraphics[width=\columnwidth,height=7.5cm]{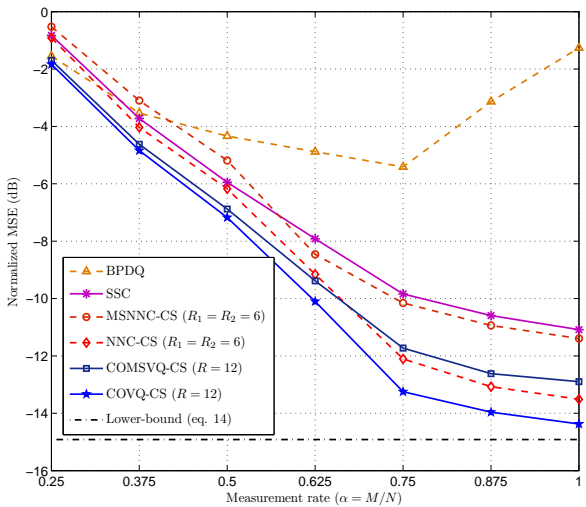}\\
  \caption{NMSE (in dB) as a function of measurement rate $\alpha = M/N$ using different quantizer design schemes. The parameters are chosen as $N=12$, $K=2$ and $R=12$ bits/vector for a noiseless channel and clean measurements.}
  \label{fig:MSE_FOM_10}
  \end{center}
\end{figure}

In our next experiment, we use larger dimension and quantization rate as ($N=32, K=3, M=20$ $(\alpha=0.625)$, $\sigma_w^2 = 0.005$) and noiseless channel. In \figref{fig:MSE_rate}, we plot the NMSE of low-complexity quantizer designs, i.e., COMSVQ-CS, MSNNC-CS, SSC and BPDQ, along with the lower-bound \eqref{eq:lower-bound SSC} by varying total quantization rate $R$. The mutual coherence $\mu$ in the lower-bound \eqref{eq:lower-bound SSC} is computed by \eqref{eq:mutual co} (here, $\mu = 0.5755$). For implementing multi-stage quantizers, we assume two stages with $R_1 = R_2$. Also, BPDQ parameters are the same as those of the previous simulation study. At low to moderate quantization rate regimes, the COMSVQ-CS outperforms other techniques; for example, at $R=20$ bits/vector, it has almost $2$ dB performance gain over MSNNC-CS and SSC, and $8$ dB performance gain over BPDQ. The performance of SSC differs much at low to high quantization rates: although its performance is very poor at low to moderate rates, the performance reaches that of the COMSVQ-CS at high rates since the SSC requires high rates to perform well. Note that, at low quantization rates, the performance of SSC is poor due to the reason that its design is based on scalar quantization of reconstructed source vector at the encoder. Whereas, the COMSVQ-CS and MSNNC-CS schemes provide better performance since, in their designs, reconstructed source vector or  CS measurement vector are vector-quantized. Naturally, this performance gain is achieved at the expense of higher encoding complexity. Note that all schemes attain a MSE floor ultimately due to the additive noise which is reflected from the lower-bound as well. In particular, at very high quantization rates, the performance of SSC approaches to that of the COMSVQ-CS, and finally converges to the CS reconstruction MSE, denoted by $D_{cs}$, which is also aligned with our findings in  \remref{rem:MSE floor}. Our calculations show that the MSE floor, i.e., the value of $D_{cs}$, is approximately $-16.5$ dB. 

\begin{figure}
  \begin{center}
  \includegraphics[width=\columnwidth,height=7.5cm]{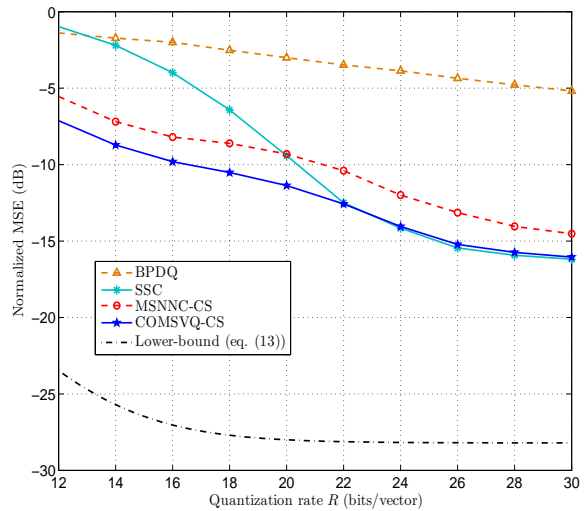}\\
  \caption{NMSE (in dB) as a function of total quantization rate $R$ (in bits/vector). Simulation parameters are chosen as $N=32$, $K=3$, $M=20$ for noiseless channel and $\sigma_w^2 = 0.005$.}
  \label{fig:MSE_rate}
  \end{center}
\end{figure}

In our final experiments, we consider the effect of channel noise on the performance of the proposed JSCC schemes, and we also compare them with separate source-channel coding schemes. In \figref{fig:MSE_eps1} and \figref{fig:MSE_eps}, we quantify the performance as a function of channel bit cross-over probability $\epsilon$, respectively, for two parameter sets: ($N \!=\! 12$, $K\!=\!2$, $M \!=\!9$ ($\alpha\!=\!0.75$), $R\!=\!15$ bits/vector, $\sigma_w^2\!=\!0$) and ($N \!= \! 32$, $K\!=\!3$, $M \!=\!20$ ($\alpha\!=\!0.625$), $R\!=\!20$ bits/vector, $\sigma_w^2\!=\!0$). In \figref{fig:MSE_eps1}, we observe that the proposed designs, i.e., COVQ-CS and COMSVQ-CS (with $R_1\!=\!8$ and $R_2\!=\!7$ bits/vector) always outperform other schemes. The curves labeled by `SSC-BCH' and `BPDQ-BCH', respectively, consist of twelve-dimensional 11-bit encoded bits using SSC and uniform quantization, followed by ($15,11$) BCH codes (this rate is experimentally tested to obtain the best performance among BCH rate allocations.). Note that channel coding rates are chosen in order to have a fair comparison (in terms of same delay) among JSCC schemes (COVQ-CS and COMSVQ-CS) and separate source-channel coding methods (SSC-BCH and BPDQ-BCH). We observe from \figref{fig:MSE_eps1} that using separate channel coding, the performance of the BPDQ is still poor. It can be also seen that the SSC, even equipped with channel coding, is highly susceptible to channel noise since an error in receiving the support set may detrimentally degrade the performance. Therefore, the MSE increases more rapidly as compared to COVQ-CS, COMSVQ-CS and BPDQ-BCH. We have also tested the performance of SSC with (15,11) Hamming codes which provides almost the same performance as that of the SSC-BCH. Using joint source-channel codes in the proposed designs enhances the performance and provide robustness, particularly at high channel noise. For example, in \figref{fig:MSE_eps1}, the performance gain of COVQ-CS over the SSC is almost $4$ dB, when the channel is highly noisy ($\epsilon \!=\! 0.05$). While the COMSVQ-CS and SSC have (almost) the same encoding complexity order, the performance gain of COMSVQ-CS over SSC-BCH is more than $2$ dB at $\epsilon \!= \!0.05$. It should be mentioned that the gap between the COVQ-CS and the lower-bound is due to CS reconstruction distortion (low number of measurements) as well as finite length of source-channel codes.

Since the SSC is quite sensitive to error in received support set, in \figref{fig:MSE_eps}, we also show the performance of SSC when the reconstructed support set is transmitted without loss, and the non-zero coefficients are encoded using ($7,4$) BCH codes. This scheme is marked by `SSC-coded' in \figref{fig:MSE_eps}, and, indeed, is an \textit{ideal} coding scheme since the support set may not be transmitted losslessly over a noisy channel, in practice. The proposed COMSVQ-CS design method not only outperforms the ideal separate source-channel coding scheme (SSC-coded), but also the other JSCC scheme (MSNNC-CS). It can be also seen as channel condition degrades, the COMSVQ-CS curve increases with the same slope as that of the lower-bound.
\begin{figure}
  \begin{center}
  \includegraphics[width=\columnwidth,height=7.5cm]{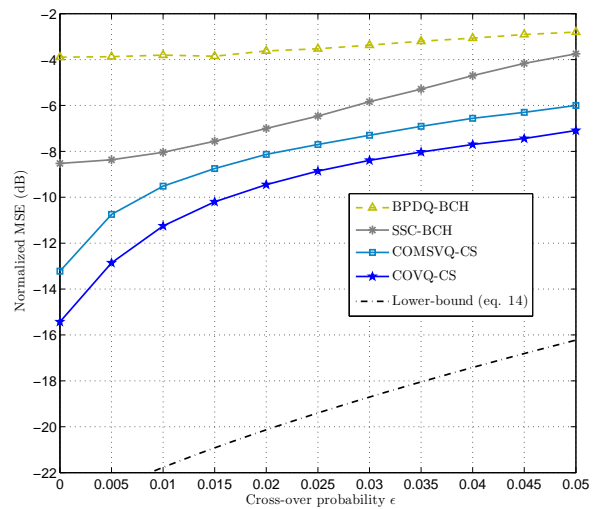}\\
  \caption{NMSE (in dB) as a function of channel bit cross-over probability  $\epsilon$. Simulation parameters are chosen as $N=12$, $K=2$, $M=9$ and $R=15$ bits/vector for clean CS measurements.}
  \label{fig:MSE_eps1}
  \end{center}
\end{figure}

\begin{figure}
  \begin{center}
  \includegraphics[width=\columnwidth,height=7.5cm]{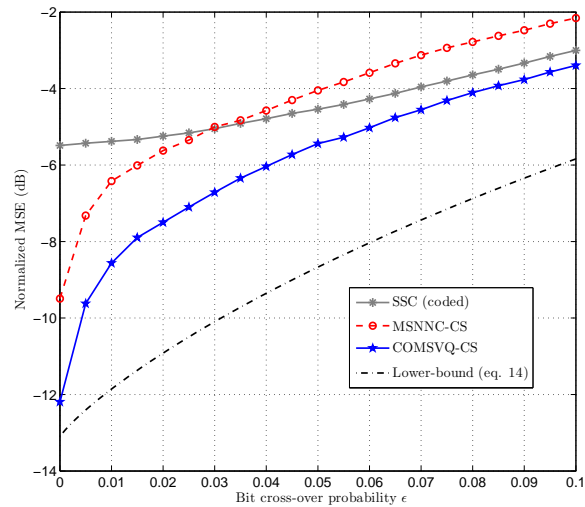}\\
  \caption{NMSE (in dB) as a function of channel bit cross-over probability $\epsilon$. Simulation parameters are chosen as $N=32$, $K=3$, $M=20$ and $R=20$ bits/vector for clean CS measurements.}
  \label{fig:MSE_eps}
  \end{center}
\end{figure}

As a final remark, we mention that the VQ, in general, is known to be theoretically the optimal block coding strategy. However, its computational and memory complexity is an issue, which has been addressed using multi-stage VQ in the current work. In our simulations, the implementation of COVQ-CS and COMSVQ-CS (using two stages) might not be performed beyond $R=12$ bits/vector for a dimension $N=12$ (or slightly more), and $R=30$ for a dimension $N=32$, respectively. If an implementation of a VQ for CS measurements of sparse sources with close to real-life dimensions (e.g., $N=256$ or even higher) is desired, one needs to consider a VQ with multiple stages (more than two). Another alternative is to use VQ for a high-dimensional source vector by segmenting the source into different small patches of information, and use the segmented patches for training.\footnote[1]{In the context of this paper, the source can be thought of as the sparse representation (e.g., wavelet coefficients, etc.) of an image.} This approach, for example, is used in vector quantization of images \cite[Chapter 11]{91:Gersho} or in image denoising \cite{14:Turek}.

\section{Conclusion} \label{sec:conclusion}
We have developed optimum joint source-channel vector quantization schemes for CS measurements. We have derived necessary conditions for optimality of VQ encoder-decoder pair with respect to end-to-end MSE. One interesting result of the optimal conditions is that the CS reconstruction should be performed MMSE-wise at the encoder side rather than the decoder side. We have also provided a theoretical lower-bound on the MSE performance based on the fact that the end-to-end MSE can be decomposed into CS reconstruction MSE and quantized transmission MSE without loss of optimality. Using the resulting optimal conditions, we have proposed a practical encoder-decoder design through an iterative algorithm referred to as COVQ-CS. Moreover, the encoding complexity of VQ was addressed using the MSVQ where we have approximated the necessary optimal conditions by applying a multi-stage structure which has led to the design of COMSVQ-CS. Numerical results show promising performance of the proposed designs with respect to relevant methods in literature.

\appendix

Note that the end-to-end MSE, $\mathbb{E}[\|\mathbf{X} - \widehat{\mathbf{X}}\|_2^2]$, of the system of \figref{fig:diagram_VQ} is always larger than the MSE of a system with a priori known (oracle) support set $\mathcal{S}$ under the same assumptions. Let the RV's $\mathbf{X}|_\mathcal{S},\widetilde{\mathbf{X}}|_\mathcal{S}, \widehat{\mathbf{X}}|_\mathcal{S} \in \mathbb{R}^N$, respectively, denote the source vector, the MMSE estimation of the source given measurements, and the final reconstruction vector given the known support set $\mathcal{S}$. Therefore, we have
\begin{equation} \label{eq:decompostion 2}
\begin{aligned}
    \mathbb{E}[\|\mathbf{X} - \widehat{\mathbf{X}}\|_2^2] &\geq \mathbb{E}[\|\mathbf{X}|_\mathcal{S} - \widehat{\mathbf{X}}|_\mathcal{S}\|_2^2]& \\
    &=\mathbb{E} [\|\mathbf{X}|_\mathcal{S} -  \widetilde{\mathbf{X}}|_\mathcal{S}\|_2^2] + \mathbb{E}[\|\widetilde{\mathbf{X}}|_\mathcal{S} - \widehat{\mathbf{X}}|_\mathcal{S}\|_2^2],&
\end{aligned}
\end{equation}
where, the equality in \eqref{eq:decompostion 2} follows from the same reasoning as that of \eqref{eq:MSE COVQ}.\footnote[2]{We drop the dependency of $\widetilde{\mathbf{X}}|_\mathcal{S}$ on $\mathbf{Y}$ for simplicity of notation.}  Let us first develop a lower-bound on $\mathbb{E} [\|\mathbf{X}|_\mathcal{S} -  \widetilde{\mathbf{X}}|_\mathcal{S}\|_2^2]$. Defining $\mathbf{\Phi}_\mathcal{S} \in \mathbb{R}^{M \times K}$ as a sub-matrix of $\mathbf{\Phi}$ formed by choosing its columns indexed by the elements of $\mathcal{S}$, then for a single realization of $\mathcal{S}$, we have

\begin{equation} \label{eq:cs dist}
\begin{aligned}
       &\mathbb{E} [\|\mathbf{X}|_\mathcal{S}   - \widetilde{\mathbf{X}}|_\mathcal{S}\|_2^2]
       \stackrel{(a)}{=} \textrm{Tr}\left\{\left(\mathbf{I}_K + \frac{1}{\sigma_w^2} \mathbf{\Phi}_\mathcal{S}^\top \mathbf{\Phi}_\mathcal{S}\right)^{-1}\right\} & \\
       &\stackrel{(b)}{\geq} \frac{K^2}{K + \frac{1}{\sigma_w^2} \textrm{Tr}\{\mathbf{\Phi}_\mathcal{S}^\top \mathbf{\Phi}_\mathcal{S}\}}
       \geq \frac{K}{1 + \frac{1}{\sigma_w^2} \lambda_{\max}\{\mathbf{\Phi}_\mathcal{S}^\top \mathbf{\Phi}_\mathcal{S}\}} & \\
        &\stackrel{(c)}{\geq} \frac{K}{1 + \frac{1}{\sigma_w^2} (1 + (K+1)\mu)},&
\end{aligned}
\end{equation}
where in $(a)$, we use the MMSE estimation error of a Gaussian source $\mathbf{X_\mathcal{S}} \sim \mathcal{N}(\mathbf{0},\mathbf{I}_K)$ (elements of $\mathbf{X}|_\mathcal{S}$ within the support set) in white Gaussian measurement noise $\mathbf{W} \sim \mathcal{N}(\mathbf{0}, \sigma_w^2\mathbf{I}_M)$ (see \cite[Theorem 11.1]{93:Kay} for details). Also, $(b)$ follows from the fact that for a given positive-definite matrix $\mathbf{B} \in \mathbb{R}^{K\times K}$, $\textrm{Tr}\{\mathbf{B}\} \textrm{Tr}\{\mathbf{B}^{-1}\} \geq K^2$ which can be shown using the Cauchy-Schwarz inequality (see e.g. \cite[Lemma 2]{03:Shengli}). Further, $(c)$ holds since all eigenvalues of $\mathbf{\Phi}_\mathcal{S}^\top\mathbf{\Phi}_\mathcal{S}$ are upper-bounded by $1 + (K+1)\mu$  using Gershgorin disc theorem \cite[Theorem 8.1.3]{96:Golub}. Now, since the oracle support set is drawn uniformly at random from all ${N \choose K}$ possibilities, the final inequality in \eqref{eq:cs dist} is the lower-bound for all realizations of $\mathcal{S}$.

Next, we develop a lower-bound on the quantized transmission distortion $\mathbb{E}[\|\widetilde{\mathbf{X}}|_\mathcal{S} - \widehat{\mathbf{X}}|_\mathcal{S}\|_2^2]$.
It should be noted that the elements of $\widetilde{\mathbf{X}}|_\mathcal{S}$ within the known support set, denoted by $\widetilde{\mathbf{X}}_\mathcal{S} \in \mathbb{R}^K$, are Gaussian with the covariance matrix (see \cite[Theorem 10.3]{93:Kay})
\begin{equation} \label{eq:mean cov}
\begin{aligned}
    &\text{cov}[\widetilde{\mathbf{X}}_\mathcal{S}] = \left(\mathbf{I}_K + \frac{1}{\sigma_w^2} \mathbf{\Phi}_\mathcal{S}^\top \mathbf{\Phi}_\mathcal{S}\right)^{-1}.&
\end{aligned}
\end{equation}
In order to find the minimum distortion, or distortion-rate function, caused by quantization of a sparse source $\widetilde{\mathbf{X}}|_\mathcal{S}$, a natural approach is to let the quantizer encoder first encode the support set elements using $\log_2 {N \choose K}$ bits (since the elements of $\mathcal{S}$ are drawn uniformly) which can be received without loss at the decoder, and then encode the correlated Gaussian vector $\widetilde{\mathbf{X}}_{\mathcal{S}}$ using $R - \log_2 {N \choose K}$ bits. It is shown in \cite{12:Weidmann} that the distortion rate function of this \textit{splitting approach} coincides with the distortion-rate function for a sparse source (with Gaussian non-zero coefficients and uniformly distributed sparsity pattern) asymptotically (in quantization rate $R$ with $R \gg \log_2 {N \choose K}$). Then, it follows that
\begin{equation} \label{eq:quant dist}
\begin{aligned}
    \mathbb{E}[\|\widetilde{\mathbf{X}}|_\mathcal{S} - \widehat{\mathbf{X}}|_\mathcal{S}\|_2^2] &\geq c 2^{-2C\left(\frac{R  - \log_2 {N \choose K}}{K}\right)} \hspace{0.1cm} \text{det}\left(\text{cov}[\widetilde{\mathbf{X}}_{\mathcal{S}}]\right)^{\frac{1}{K}}, &
\end{aligned}
\end{equation}
where $c = 2\left(\frac{K}{2} \Gamma\left(\frac{K}{2}\right) \right)^{\frac{2}{K}} \left(\frac{K+2}{K}\right)^{\frac{K}{2}}$. The right-hand side in \eqref{eq:quant dist} is indeed the distortion-rate function of the correlated Gaussian source $\widetilde{\mathbf{X}}_\mathcal{S}$ \cite{80:Yamada} incurred by transmission over the DMC with capacity $C$ (see, e.g. source-channel separation theorem \cite[Chapter 7]{06:Cover}). Further, we have
\begin{equation} \label{eq:LB det}
\begin{aligned}
    &\text{det}\left(\text{cov}[\widetilde{\mathbf{X}}_{\mathcal{S}}]\right)^{\frac{1}{K}} 
       = \frac{1}{\prod_{k=1}^K \left(1 + \frac{1}{\sigma_w^2} \lambda_k \left(\mathbf{\Phi}_{\mathcal{S}}^\top \mathbf{\Phi}_{\mathcal{S}}\right)\right)^{\frac{1}{K}}} & \\
    &\geq \frac{1}{1 + \frac{1}{\sigma_w^2} \lambda_{\max} \left(\mathbf{\Phi}_{\mathcal{S}}^\top \mathbf{\Phi}_{\mathcal{S}}\right)}
    \stackrel{(a)}{\geq} \frac{1}{1 + \frac{1}{\sigma_w^2} (1 + (K+1)\mu)},&
\end{aligned}
\end{equation}
where in $(a)$ we use the fact that all eigenvalues of $\mathbf{\Phi}_\mathcal{S}^\top\mathbf{\Phi}_\mathcal{S}$ are upper-bounded by $1 + (K+1)\mu$  using Gershgorin disc theorem. Combining \eqref{eq:LB det} with \eqref{eq:quant dist}, \eqref{eq:cs dist} and \eqref{eq:decompostion 2} concludes the proof. 

\bibliographystyle{IEEEtran}
\bibliography{IEEEfull,bibliokthPasha}

\vspace{-1cm}
\begin{IEEEbiography}[{\includegraphics[width=0.82in,height=1.25in]{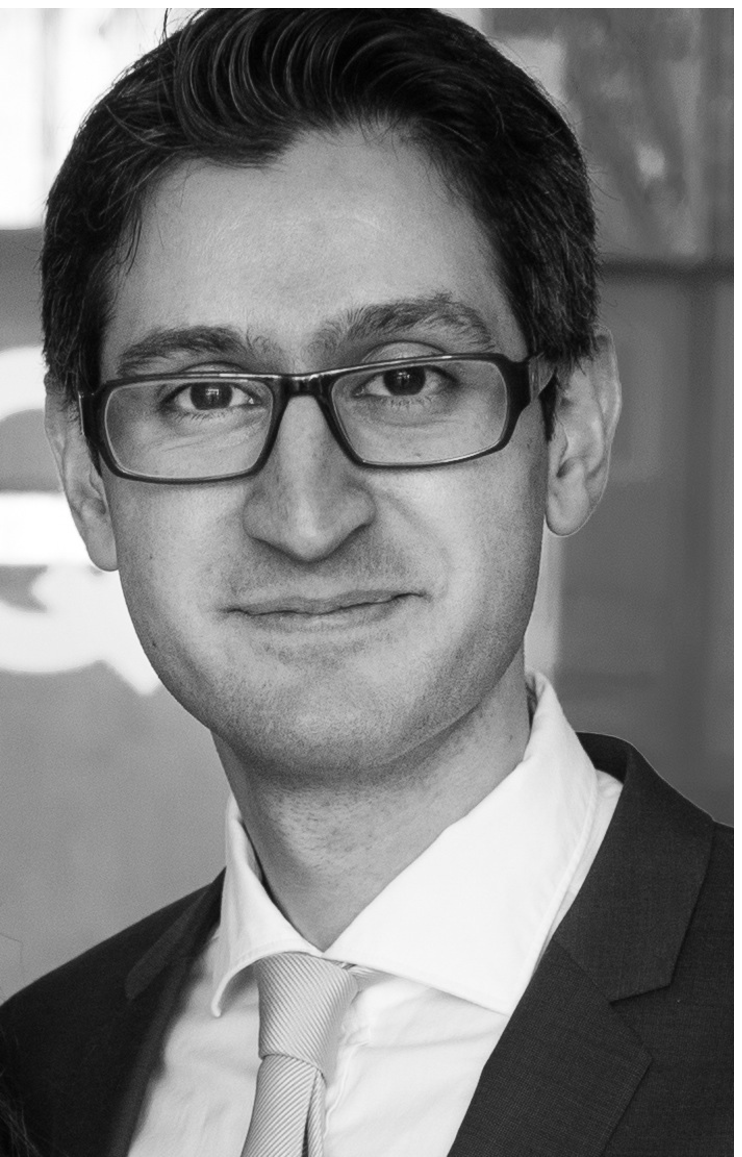}}]{Amirpasha Shirazinia}
(S'08) received the BSc and MSc degrees, in communication systems, from Ferdowsi University of Mashhad, and Amirkabir University of Technology (Tehran Polytechnic), Iran, in 2005 and 2008, respectively. He received his PhD degree in Telecommunications from KTH--Royal Institute of Technology,  Department of Communication Theory, Stockholm, Sweden, in 2014. Dr. Shirazinia is currently a post-doctoral researcher at Uppsala University, Sweden. His research interests include joint source-channel coding, statistical signal processing, compressed sensing and sensor networks. 
\end{IEEEbiography}
\vspace{-1.25cm}
\begin{IEEEbiography}[{\includegraphics[width=1in,height=1.3in,clip,keepaspectratio]{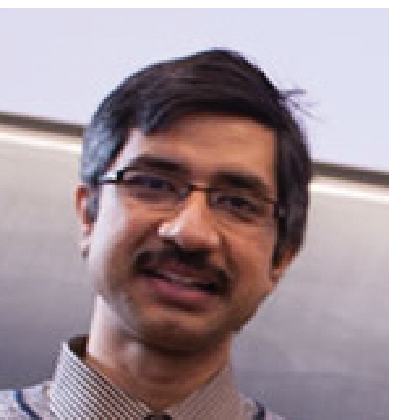}}]{Saikat Chatterjee}
is a researcher jointly with the Communication Theory and the Signal Processing Divisions, School of Electrical Engineering, KTH-Royal Institute of Technology, Sweden. He was also with the Sound and Image Processing Division at the same institution as a post-doctoral fellow for one year. Before moving to Sweden, he received Ph.D. degree in 2009 from Indian Institute of Science, India. He was a co-author of the paper that won the best student paper award at ICASSP 2010. His current research interests are source coding, speech and audio processing, estimation and detection, sparse signal processing, compressive sensing, wireless communications and computational biology.
\end{IEEEbiography}

\vspace{-1.25cm}
\begin{IEEEbiography}[{\includegraphics[width=1in,height=1.25in,clip,keepaspectratio]{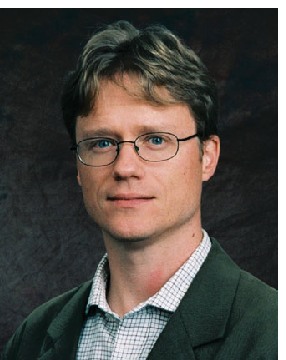}}]{Mikael Skoglund}
(S'93-M'97-SM'04) received the Ph.D.~degree in 1997 from Chalmers University of Technology, Sweden.  In 1997, he joined the KTH-Royal Institute of Technology, Stockholm, Sweden, where he was appointed to the Chair in Communication Theory in 2003.  At KTH, he heads the Communication Theory Division and he is the Assistant Dean for Electrical Engineering. He is also a founding faculty member of the ACCESS Linnaeus Center and director for the Center Graduate School.

Dr.~Skoglund has worked on problems in source-channel coding, coding and transmission for wireless communications, Shannon theory and statistical signal processing. He has authored and co-authored more than 100 journal and 250 conference papers, and he holds six patents.

Dr.~Skoglund has served on numerous technical program committees for IEEE sponsored conferences (including ISIT and ITW). During 2003--08 he was an associate editor with the IEEE Transactions on Communications and during 2008--12 he was on the editorial board for the IEEE Transactions on Information Theory.
\end{IEEEbiography}

\end{document}